\documentclass[journal,onecolumn]{IEEEtran}
\usepackage{amsthm}
\usepackage{amsmath}
\usepackage{mathrsfs}
\usepackage{amssymb} 
\usepackage{bbm,dsfont} 
\usepackage{multirow} 
\usepackage{units}
\usepackage{stmaryrd}   
\usepackage{verbatim}
\usepackage{enumerate} 
\usepackage[table]{xcolor}
\usepackage{slashbox} 
\usepackage{ wasysym }
\usepackage{tikz}
\usetikzlibrary{calc}
\usepackage{tabularx}

\usepackage{lipsum}
\usepackage{multicol}
\usepackage{float}

\usepackage{ifthen}

\usepackage{xcolor}
\usepackage{hyperref} 
\usepackage
{hyperref} 
\hypersetup{
    colorlinks,
    linkcolor={blue!80!black},
    citecolor={green!50!black},
    urlcolor={blue!80!black}
}

\newcommand\blfootnote[1]{%
  \begingroup
  \renewcommand\thefootnote{}\footnote{#1}%
  \addtocounter{footnote}{-1}%
  \endgroup
}

\usepackage{graphicx}
\graphicspath{{images/}}

\title{The Arbitrarily Varying Relay Channel}
\author{\IEEEauthorblockN{Uzi Pereg and Yossef Steinberg}
\IEEEauthorblockA{Department of Electrical Engineering\\
Technion, Haifa 32000, Israel.\\
Email: {\tt uzipereg@campus.technion.ac.il}, {\tt ysteinbe@ee.technion.ac.il}
 }
}

\usepackage[backend=biber,sorting=nyt,maxbibnames=10]{biblatex}
\ifdefined\bibstar\else\newcommand{\bibstar}[1]{}\fi
\renewbibmacro{in:}{} 										

\definecolor{light-gray}{gray}{0.8}

\definecolor{dark-gray}{gray}{0.3}
\usepackage{accents}
\newlength{\dhatheight}

\newcommand{\bieee}{\begin{IEEEeqnarray}{rCl}}
\newcommand{\eieee}{\end{IEEEeqnarray}}
\newcommand{\prob}[1]{\Pr\left(#1\right)}
\newcommand{\given}{\mid}
\newcommand{\cprob}[2]{\Pr\left(#1\given #2\right)}
\newcommand{\E}{\mathbb{E}}
\newcommand{\var}{\mathbb{V}\mathrm{ar}}
\newcommand{\eps}{\varepsilon}

\newcommand{\norm}[1]{\left\lVert#1\right\rVert}

\newcommand{\ie}{\emph{i.e.} }
\newcommand{\eg}{\emph{e.g.} }
\newcommand{\etal}{\emph{et al.} }
\newcommand{\cf}{\emph{cf.} }

\newcommand{\xvec}{\mathbf{x}}
\newcommand{\yvec}{\mathbf{y}}
\newcommand{\zvec}{\mathbf{z}}
\newcommand{\fvec}{\mathbf{f}}

\newcommand{\svec}{\mathbf{s}}
\newcommand{\avec}{\mathbf{a}}

\newcommand{\cvec}{\mathbf{c}}
	
\newcommand{\vvec}{\mathbf{v}}	
\newcommand{\Xvec}{\mathbf{X}}
\newcommand{\Yvec}{\mathbf{Y}}
\newcommand{\Svec}{\mathbf{S}}
		
\newcommand{\tm}{\widetilde{m}}	
\newcommand{\tM}{\widetilde{M}}																				
\newcommand{\tq}{\widetilde{q}}

\newcommand{\tx}{\tilde{x}}

\newcommand{\ty}{\tilde{y}}

\newcommand{\tf}{\widetilde{f}}
\newcommand{\tg}{\widetilde{g}}

\newcommand{\gnu}{\tg}
\newcommand{\bR}{\tR}
\newcommand{\oS}{\overline{S}}

\newcommand{\oq}{\overline{q}}

\newcommand{\tR}{\widetilde{R}}

\newcommand{\ha}{\hat{a}}
\newcommand{\hj}{\hat{j}}
\newcommand{\hk}{\hat{k}}
\newcommand{\hm}{\hat{m}}
\newcommand{\hy}{\hat{y}}
\newcommand{\hq}{\widehat{q}}

\newcommand{\hgamma}{\widehat{\gamma}}
\newcommand{\hJ}{\hat{J}}

\newcommand{\hP}{\hat{P}}
\newcommand{\hM}{\hat{M}}
\newcommand{\hY}{\hat{Y}}

\newcommand{\Aset}{\mathcal{A}}

\newcommand{\Dset}{\mathcal{D}}
\newcommand{\Fset}{\mathcal{F}}

\newcommand{\Jset}{\mathcal{J}}
\newcommand{\Uset}{\mathcal{U}}

\newcommand{\Qset}{\mathcal{Q}}

\newcommand{\Sset}{\mathcal{S}}
\newcommand{\Wset}{\mathcal{W}}
\newcommand{\Xset}{\mathcal{X}}
\newcommand{\Yset}{\mathcal{Y}}
\newcommand{\Zset}{\mathcal{Z}}
\newcommand{\Eset}{\mathcal{E}}
\newcommand{\markovC}[1]{%
\begin{tikzpicture}[#1]%
\draw (0,0.3ex) -- (1ex,0.3ex);%
\draw (0.5ex,0.3ex) circle (0.2ex);
\draw[white] (0.2ex,0) -- (0.5ex,0);%
\end{tikzpicture}%
}
\newcommand{\Cbar}{\markovC{scale=2}}

\theoremstyle{remark}	\newtheorem{theorem}{Theorem}
\theoremstyle{remark}	\newtheorem{lemma}[theorem]{Lemma}
\theoremstyle{remark}	\newtheorem{coro}[theorem]{Corollary}
\theoremstyle{remark} \newtheorem{definition}{Definition}
\theoremstyle{remark} 
\theoremstyle{remark} \newtheorem{example}{Example}

\newcommand{\avc}{\Wset}																		
\newcommand{\opC}{\mathbb{C}}																
\newcommand{\inR}{\mathsf{R}}

\newcommand{\pSpace}{\mathcal{P}}														

\newcommand{\dM}{\mathsf{M}}															 	

\newcommand{\enc}{f}																				
\newcommand{\renc}{\mathrm{f}}																				
\newcommand{\dec}{g}																			 	
\newcommand{\code}{\mathscr{C}}															
\newcommand{\gcode}{\mathscr{C}^{\,\Gamma}}									

\newcommand{\cerr}{P_{e|s^n}^{(n)}}													
\newcommand{\err}{P_e^{(n)}}															
\newcommand{\plimit}{\Omega}																			
\newcommand{\tset}{\Aset^{\delta}}													
\newcommand{\qn}{q}
\newcommand{\tQ}{\hat{\Qset}_n}														

\newcommand{\encn}{f^n}																			

\newcommand{\rstarC}{																			  
\, \hspace{-0.3cm} \text{ $$ \mbox{  
\hspace{-0.1cm} 
\small $\star$   
} $$ }
\hspace{-0.25cm}}

\newcommand{\rCav}{\opC^{\rstarC}\hspace{-0.1cm}(\avc)}

\newcommand{\rc}{W_{Y,Y_1|X,X_1,S}}													
\newcommand{\nRC}{W_{Y^n,Y_1^n|X^n,X_1^n,S^n}} 									
\newcommand{\tRYset}{\mathcal{L}}
\newcommand{\RYcompound}{\tRYset^\Qset} 
\newcommand{\RYcompoundP}{\tRYset^{\pSpace(\Sset)}} 
\newcommand{\avrc}{\tRYset}																			
\newcommand{\RYopC}{\mathbb{C}}																
									
\newcommand{\RYrCcompound}{\RYopC^{\rstarC}\hspace{-0.1cm}(\RYcompound)}
\newcommand{\RYrCcompoundP}{\RYopC^{\rstarC}\hspace{-0.1cm}(\RYcompoundP)}

\newcommand{\RYrCav}{\RYopC^{\rstarC}\hspace{-0.1cm}(\avrc)}  

\newcommand{\RYrICav}{\inR_{CS}^{\rstarC}(\avrc)} 

\newcommand{\RYCcompound}{\RYopC(\RYcompound)}
\newcommand{\RYCavc}{\RYopC(\avrc)}
\newcommand{\RYICcompound}{\inR_{CS}(\RYcompound)}
\newcommand{\RYdIRcompound}{\inR_{DF}(\RYcompound)} 
\newcommand{\RYcIRcompound}{\inR_{comp}(\RYcompound)}

\newcommand{\RYdIRavc}{\inR_{DF}^{\rstarC}(\avrc)}

\newcommand{\davrc}{\avrc_2}

\newcommand{\dRYdIRavc}{\inR_{DF}^{\rstarC}(\davrc)}

\begin{document}
\maketitle

{}

\begin{abstract} 
We study the arbitrarily varying relay channel, and  
  establish the cutset bound and partial decode-forward bound  on the random code capacity. 
We further determine the random code capacity for special cases. Then, we consider conditions under which the deterministic code capacity is determined as well.
\end{abstract}


\blfootnote{
This work was supported by the Israel Science Foundation (grant No. 1285/16).
}

\section{Introduction}

The relay channel was first introduced by van der Meulen \cite{vanderMeulen:71p}
 to describe  point to point communication with the help of a relay, which
 receives a noisy version of the transmitter signal and transmits a signal of its own to the destination receiver, in a strictly causal manner. The capacity of the relay channel is not known in general, however,
 Cover and El Gamal established
the cutset upper bound, the  decode-forward lower bound, and the partial decode-forward lower bound 
\cite{CoverElGamal:79p}. 
It was also shown in \cite{CoverElGamal:79p} that for the reversely degraded relay channel, direct transmission is capacity achieving.  For the degraded relay channel, the decode-forward  bound and the cutset bound coincide, thus characterizing the capacity for this model \cite{CoverElGamal:79p}.  In general, the partial decode-forward lower bound is tighter than both direct transmission and  decode-forward lower bounds.
 El Gamal and Zahedi \cite{ElGamalZahedi:05p} determined the capacity of the relay channel with orthogonal sender components, by showing that the partial decode-forward bound and cutset bound coincide.  

 In practice, the channel statistics are not necessarily known in exact, and they may even change over time.
The arbitrarily varying channel (AVC) is an appropriate model to describe such a situation 
 \cite{BBT:60p}. Considering the AVC without a relay, 
Blackwell \etal  determined  the random code capacity  \cite{BBT:60p}, \ie the capacity achieved by stochastic-encoder stochastic-decoder coding schemes with common randomness. It was also demonstrated in  \cite{BBT:60p}  that the random code capacity is not necessarily  achievable using deterministic codes. 
  A well-known result by Ahlswede \cite{Ahlswede:78p} is the dichotomy property presented by the AVC.
	Specifically,  the deterministic code capacity either equals the random code capacity or else, it is zero. 
Subsequently, Ericson \cite{Ericson:85p} and Csisz{\'{a}}r and Narayan \cite{CsiszarNarayan:88p}
 established a simple single-letter condition, namely non-symmetrizability, which is both necessary and sufficient for the capacity to be positive. 

In this work, we study the arbitrarily varying relay channel (AVRC), which combines the previous models, \ie the relay channel and the AVC. 
In the analysis, we incorporate the block Markov coding schemes of \cite{CoverElGamal:79p} in Ahlswede's  Robustification and Elimination Techniques \cite{Ahlswede:78p,
Ahlswede:86p}. 
%
We establish the cutset upper bound and the full/partial decode-forward lower bound 
 on the random code capacity of the AVRC. 
We determine the random code capacity for special cases of the degraded AVRC, the reversely degraded AVRC, and the AVRC with orthogonal sender components.  Then, we give extended symmetrizability conditions under which the deterministic code capacity coincides with the random code capacity. 
We show by example that 
the deterministic code capacity can be strictly lower than the random code capacity of the AVRC.
We also give generalized symmetrizability conditions under which the deterministic code capacity coincides with the random code capacity, and conditions under which it is zero.  

\section{Definitions}
\label{sec:def}
\subsection{Notation}
\label{sec:notation}
We use the following notation conventions throughout. 
Calligraphic letters $\Xset,\Sset,\Yset,...$ are used for finite sets.
Lowercase letters $x,s,y,\ldots$  stand for constants and values of random variables, and uppercase letters $X,S,Y,\ldots$ stand for random variables.  
 The distribution of a random variable $X$ is specified by a probability mass function (pmf) 
	$P_X(x)=p(x)$ over a finite set $\Xset$. The set of all pmfs over $\Xset$ is denoted by $\pSpace(\Xset)$. 
 We use $x^j=(x_1,x_{2},\ldots,x_j)$ to denote  a sequence of letters from $\Xset$. 
 A random sequence $X^n$ and its distribution $P_{X^n}(x^n)=p(x^n)$ are defined accordingly. 
For a pair of integers $i$ and $j$, $1\leq i\leq j$, we define the discrete interval $[i:j]=\{i,i+1,\ldots,j \}$. The notation $\xvec=(x_1,x_{2},\ldots,x_n)$ is used 
when it is understood from the context that the length of the sequence is $n$, and  the $\ell^2$-norm of $\xvec$ is denoted  by $\norm{\xvec}$.

\subsection{Channel Description}
A state-dependent discrete memoryless relay channel 
$(\Xset,\Xset_1,\Sset,\rc,\Yset,\Yset_1)$ consists of five sets, $\Xset$, $\Xset_1$, $\Sset$, $\Yset$ and $\Yset_1$, and a collection of conditional pmfs $\rc$.  The sets
 stand for the input alphabet, the relay transmission alphabet, the state alphabet, the output alphabet, and the relay input alphabet, respectively.  The alphabets are assumed to be finite, unless explicitly said otherwise.
%
 The channel is memoryless without feedback, and therefore  
\begin{align}
W_{Y^n,Y_1^n|X^n,X_1^n,S^n}(y^n,y_1^n|x^n,x_1^n,s^n)= \prod_{i=1}^n \rc(y_{i},y_{1,i}|x_i,x_{1,i},s_i) \,.
 \end{align}
Communication over a relay channel is depicted in Figure~\ref{fig:RC}.
Following \cite{SigurjonssonKim:05c}, a relay channel $\rc$ is called degraded if the channel can be expressed as
\begin{align}
\label{eq:RYdegraded}
\rc(y,y_1|x,x_1,s)=W_{Y_1|X,X_1,S}(y_1|x,x_1,s)p(y|y_1,x_1,s) \,,
\end{align}
and it is called reversely degraded if
\begin{align}
\label{eq:RYdegradedRev}
\rc(y,y_1|x,x_1,s)=W_{Y|X,X_1,S}(y|x,x_1,s)p(y_1|y,x_1,s) \,.
\end{align}
 

The \emph{arbitrarily varying relay channel} (AVRC) is a discrete memoryless relay channel $(\Xset,\Xset_1,\Sset,\rc,\Yset,\Yset_1)$  with a state sequence of unknown distribution,  not necessarily independent nor stationary. That is, $S^n\sim \qn(s^n)$ with an unknown joint pmf $\qn(s^n)$ over $\Sset^n$.
In particular, $\qn(s^n)$ can give mass $1$ to some state sequence $s^n$. 
We use the shorthand notation $\avrc=\{\rc\}$ for the AVRC, where the 
alphabets are understood from the context.

To analyze the AVRC, we consider the \emph{compound  relay channel}. 
Different models of  compound relay channels have been considered in the literature 
\cite{SimeoneGunduzShamai:09c,BehboodiPiantanida:09c}.
Here, we define the compound  relay channel as a discrete memoryless relay channel $(\Xset,\Xset_1,\Sset,$ $\rc,\Yset,\Yset_1)$ with a discrete memoryless state, where the state distribution $q(s)$ is not known in exact, but rather belongs to a family of distributions $\Qset$, with $\Qset\subseteq \pSpace(\Sset)$. That is,  $S^n\sim\prod_{i=1}^n q(s_i)$, with an unknown pmf $q\in\Qset$ over $\Sset$.
We use the shorthand notation $\RYcompound$ for the compound relay channel,
where the transition probability $\rc$ and the alphabets are understood from the context. 

In the analysis, we also use the following model. Suppose that the user transmits $B>0$ blocks of length $n$, and the jammer is entitled to use a different state distribution $q_b(s)\in\Qset$  for every block  $b\in [1:B]$, while the encoder, relay and  receiver are aware of this jamming scheme. In other words, every block is governed by a different memoryless state. 
 We refer to this channel as the block-compound relay channel, denoted by $\avrc^{\Qset\times B}$.
Although this is a toy model, it is a useful tool for the analysis of the AVRC.

\begin{center}
\begin{figure}[htb]
        \centering
        \includegraphics[scale=0.65,trim={1cm 0 0 0},clip] 
				{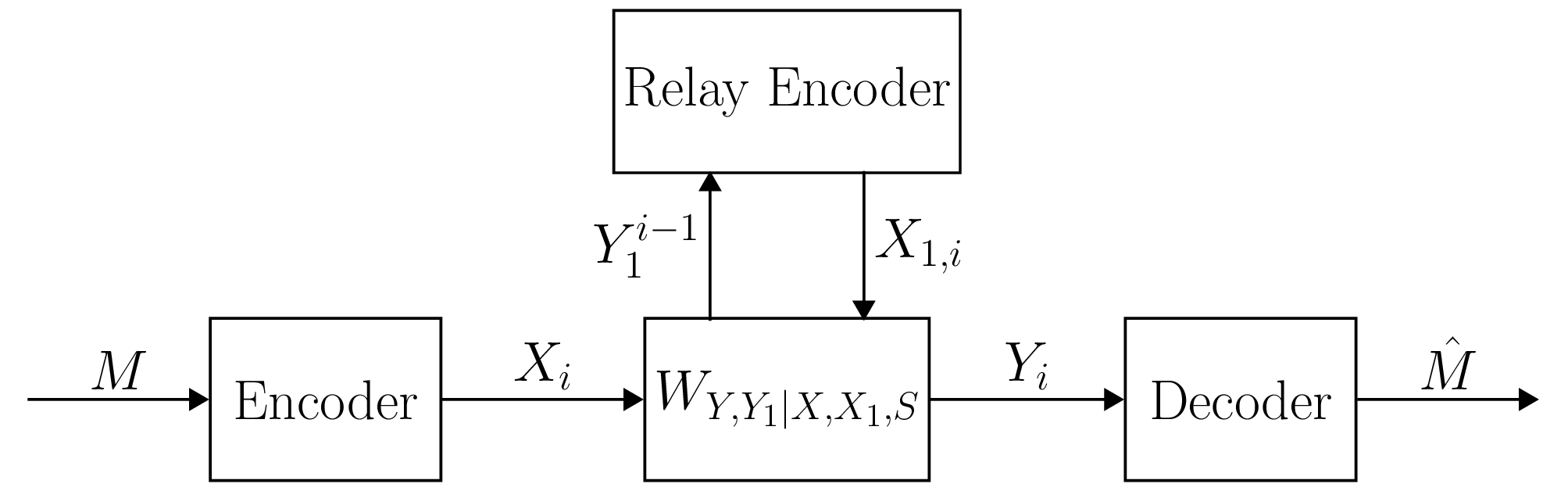}
        
\caption{Communication over the AVRC $\avrc=\{\rc\}$. 	 Given a message $M$, 
 the encoder transmits $X^n=\enc(M 
)$. 
At time $i\in [1:n]$, the relay transmits $X_{1,i}$ based on all the symbols of the past $Y_1^{i-1}$ and then receives a new symbol $Y_{1,i}$. The decoder receives the output sequence $Y^n$ and  finds an estimate of the message $\;\hM=g(Y^n)$. 
  }
\label{fig:RC}
\end{figure}
\end{center}

\subsection{Coding}
\label{subsec:RYcoding}
We introduce some preliminary definitions, starting with the definitions of a deterministic code and a random code for the AVRC $\avrc$. 
Note that in general, the term `code', unless mentioned otherwise, refers to a deterministic code.  

\begin{definition}[A code, an achievable rate and capacity]
\label{def:RYcapacity}
A $(2^{nR},n)$ code for the AVRC $\avrc$ 
 consists of the following;  a  message set $[1:2^{nR}]$,  where it is assumed throughout that $2^{nR}$ is an integer, 
an encoder $\enc:[1:2^{nR}]\rightarrow\Xset^n$,
a sequence of $n$ relaying functions $\enc_{1,i}:\Yset_1^{i-1}\rightarrow\Xset_{1,i}$,
$i\in [1:n]$, and a  decoding function $\dec: \Yset^n\rightarrow [1:2^{nR}]$.

 Given a message $m\in [1:2^{nR}]$, 
 the encoder transmits $x^n=\enc(m 
)$. 
At time $i\in [1:n]$, the relay transmits $x_{1,i}=\enc_{1,i}(y_1^{i-1})$ and then receives  $y_{1,i}$. The relay codeword is  given by $x_1^n=\encn_1(y_1^n)\triangleq \left( \enc_{1,i}(y_1^{i-1}) \right)_{i=1}^n$.
The decoder receives the output sequence $y^n$ and  finds an estimate of the message $\hm=g(y^n)$
(see Figure~\ref{fig:RC}). 
We denote the code by $\code=\left(\enc(\cdot
),\encn_1(\cdot),\dec(\cdot) \right)$.
 Define the conditional probability of error of the code $\code$ given a state sequence $s^n\in\Sset^n$ by  
\begin{align}
\label{eq:RYcerr}
\cerr(\code)=
\frac{1}{2^{nR}}\sum_{m=1}^{2^ {nR}}
\sum_{ 
\substack{
(y^n,y_1^n)\,:\; \dec(y^n)\neq m
}
} 
 \left[ \prod_{i=1}^n \rc(y_i,y_{1,i}|\enc_i(m
),f_{1,i}(y_1^{i-1}),s_i) \right] \,.
\end{align}
Now, define the average probability of error of $\code$ for some distribution $\qn(s^n)\in\pSpace(\Sset^n)$, 
\begin{align}
\err(\qn,\code)
=\sum_{s^n\in\Sset^n} \qn(s^n)\cdot\cerr(\code) \,.
\end{align}
Observe that $ \err(\qn,\code)$ is linear  in $q$, and thus continuous. 
We say that $\code$ is a 
$(2^{nR},n,\eps)$ code for the AVRC $\avrc$ if it further satisfies 
\begin{align}
\label{eq:RYerr}
 \err(\qn,\code)\leq \eps \,,\quad\text{for all $\qn(s^n)\in\pSpace(\Sset^n)\,$.} 
\end{align} 
A rate $R$ is called achievable if for every $\eps>0$ and sufficiently large $n$, there exists a  $(2^{nR},n,\eps)$ code. The operational capacity is defined as the supremum of the achievable rates and it is denoted by $\RYCavc$. 
 We use the term `capacity' referring to this operational meaning, and in some places we call it the deterministic code capacity in order to emphasize that achievability is measured with respect to  deterministic codes.  
\end{definition} 

We proceed now to define the parallel quantities when using stochastic-encoders stochastic-decoder triplets with common randomness.
The codes formed by these triplets are referred to as random codes. 

\begin{definition}[Random code]
\label{def:RYcorrC} 
A $(2^{nR},n)$ random code for the AVRC $\avrc$ consists of a collection of 
$(2^{nR},n)$ codes $\{\code_{\gamma}=(\enc_\gamma,\encn_{1,\gamma},\dec_{\gamma})\}_{\gamma\in\Gamma}$, along with a probability distribution $\mu(\gamma)$ over the code collection $\Gamma$. 
We denote such a code by $\gcode=(\mu,\Gamma,\{\code_{\gamma}\}_{\gamma\in\Gamma})$.
Analogously to the deterministic case,  a $(2^{nR},n,\eps)$ random code has the additional requirement
\begin{align}
 \err(\qn,\gcode)=\sum_{\gamma\in\Gamma} \mu(\gamma)\err(q,
\code_\gamma)\leq \eps \,,\;\text{for all $\qn(s^n)\in\pSpace(\Sset^n)$} & \,. \qquad
\end{align}
The capacity achieved by random codes is denoted by $\RYrCav$, and it 
 is referred to as the random code capacity.
\end{definition}

\section{Main Results -- General AVRC}
\label{sec:RYres}
We present our results on the compound relay channel and the AVRC. 

\subsection{The Compound Relay Channel} 
\label{sec:RYcompound}
 We establish the cutset upper bound and the partial decode-forward lower bound for the compound relay channel. 
 Consider a given compound relay channel $\RYcompound$. 
Let 
\begin{align}  
\label{eq:RYICcompound} 
\RYICcompound \triangleq& \inf_{q\in\Qset}  \max_{p(x,x_1)} 
\min \left\{ I_q(X,X_1;Y) \,,\; I_q(X;Y,Y_1|X_1)
\right\} \,,
\intertext{and}
\RYdIRcompound\triangleq&
\max_{p(u,x,x_1)} 
\min \Big\{ \inf_{q\in\Qset} I_q(U,X_1;Y)+ \inf_{q\in\Qset} I_q(X;Y|X_1,U) \,,\; \nonumber\\&
\qquad\qquad\qquad\inf_{q\in\Qset} I_{q}(U;Y_1|X_1)+ \inf_{q\in\Qset} I_q(X;Y|X_1,U)
\Big\} \,,
\label{eq:RYdIRcompound}
\end{align}
where the subscripts `$CS$' and `$DF$' stand for `cutset' and `decode-forward', respectively.
 
\begin{lemma}
\label{lemm:RYcompoundDF}
The capacity of the compound relay channel $\RYcompound$ 
 is bounded by
\begin{align}
 &\RYCcompound \geq \RYdIRcompound \,, \\
 &\RYrCcompound \leq \RYICcompound \,.
\end{align}
Specifically, if $R<\RYdIRcompound$, then there exists a $(2^{nR},n,e^{-an})$ block Markov code over $\RYcompound$ for 
 sufficiently large $n$ and some $a>0$.
\end{lemma}
The proof of Lemma~\ref{lemm:RYcompoundDF} is given in Appendix~\ref{app:RYcompoundDF}.
Observe that taking $U=\emptyset$  in (\ref{eq:RYdIRcompound}) gives the direct transmission lower bound,
\begin{align}
\RYCcompound\geq& \RYdIRcompound \geq 
\max_{p(x,x_1)} \inf_{q\in\Qset} I_q(X;Y|X_1) \,.
\label{eq:RYcompoundDirectTran}
\intertext{
Taking $U=X$  in (\ref{eq:RYdIRcompound}) results in a full decode-forward lower bound,
}
\RYCcompound\geq& \RYdIRcompound \geq \max_{p(x,x_1)} \inf_{q\in\Qset}  \min \left\{ I_q(X,X_1;Y) \,,\; 
 I_{q}(X;Y_1|X_1)
\right\} \,.
\label{eq:RYcompoundFullDF}
\end{align}
This yields the following corollary.
\begin{coro}
\label{coro:RYcompoundDeg}
Let $\RYcompound$ be a compound relay channel, where $\Qset$ is a compact convex set.
\begin{enumerate}[1)]
\item
If $\rc$ is reversely degraded, such that 
\begin{align}
\rc(y,y_1|x,x_1,s)=W_{Y|X,X_1}(y|x,x_1)W_{Y_1|Y,X_1,S}(y_1|y,x_1,s) \,,
\end{align}
 then
\begin{align}
\RYCcompound= \RYdIRcompound=\RYICcompound= \min_{q\in\Qset} \max_{p(x, x_1)}  I_q(X;Y|X_1) \,.
\end{align}
\item
If $\rc$ is degraded, such that 
\begin{align}
\rc(y,y_1|x,x_1,s)=W_{Y_1|X,X_1}(y_1|x,x_1)W_{Y|Y_1,X_1,S}(y|y_1,x_1,s) \,,
\end{align}
 then 
\begin{align}
\RYCcompound= \RYdIRcompound=\RYICcompound=  \max_{p(x,x_1)}   \min \left\{ \min_{q\in\Qset} I_q(X,X_1;Y) \,,\; 
 I(X;Y_1|X_1)
\right\} \,.
\end{align}
\end{enumerate}
\end{coro}
The proof of Corollary~\ref{coro:RYcompoundDeg} is given in Appendix~\ref{app:RYcompoundDeg}. 
%
Part 1 follows from the direct transmission and cutset bounds, (\ref{eq:RYcompoundDirectTran}) and  (\ref{eq:RYICcompound}), respectively, while  part 2 is based on the full decode-forward and cutset bounds, (\ref{eq:RYcompoundFullDF}) and  (\ref{eq:RYICcompound}), respectively. 
The following corollary is a direct consequence of Lemma~\ref{lemm:RYcompoundDF}  and it is  significant for the random code analysis of the AVRC.
\begin{coro}
\label{coro:RYcompoundDFb}
The capacity of the block-compound relay channel $\avrc^{\Qset\times B}$ 
 is bounded by
\begin{align}
 &\opC(\avrc^{\Qset\times B} ) \geq \RYdIRcompound \,, \label{eq:RYcDFb} \\
 &\opC^{\rstarC}(\avrc^{\Qset\times B}) \leq \RYICcompound \,. \label{eq:RYcCSb}
\end{align} 
Specifically, if $R<\RYdIRcompound$, then there exists a $(2^{nR},n,e^{-an})$ block Markov code over $\avrc^{\Qset\times B}$ for 
 sufficiently large $n$ and some $a>0$.
\end{coro}
The proof of Corollary~\ref{coro:RYcompoundDFb} is given in Appendix~\ref{app:RYcompoundDFb}.

\subsection{The AVRC}
We give lower and upper bounds, on the random code capacity and the deterministic code capacity, for the 
AVRC $\avrc$. 
\subsubsection{Random Code Lower and Upper Bounds}
Define 
\begin{align}
\label{eq:RYIRcompoundP}  
\RYdIRavc \triangleq \RYdIRcompound\bigg|_{\Qset=\pSpace(\Sset)} 
\,,\;
\RYrICav\triangleq \RYICcompound\bigg|_{\Qset=\pSpace(\Sset)}
\,.
\end{align}

\begin{theorem}
\label{theo:RYmain}
The random code capacity  of an AVRC $\avrc$ is bounded by
\begin{align}
\RYdIRavc\leq \RYrCav \leq \RYrICav \,.
\end{align}
\end{theorem}
The proof of Theorem~\ref{theo:RYmain} is given in Appendix~\ref{app:RYmain}. 
Together with Corollary~\ref{coro:RYcompoundDeg}, this yields another corollary.
\begin{coro}
\label{coro:RYavrcDeg}
Let $\avrc$ be an AVRC.
\begin{enumerate}[1)]
\item
If $\rc$ is reversely degraded, such that $\rc=W_{Y|X,X_1} W_{Y_1|Y,X_1,S}$,
 then
\begin{align}
\RYrCav= \RYdIRavc=\RYrICav=  \min_{q(s)} \max_{p(x,x_1)} I_q(X;Y|X_1) \,.
\end{align}
\item
If $\rc$ is degraded, such that $\rc=W_{Y_1|X,X_1} W_{Y|Y_1,X_1,S}$, then
\begin{align}
\RYrCav= \RYdIRavc=\RYrICav=  \max_{p(x,x_1)}   \min \left\{ \min_{q(s)} I_q(X,X_1;Y) \,,\; 
 I(X;Y_1|X_1)
\right\} \,.
\end{align}
\end{enumerate}
\end{coro}

Before we proceed to the deterministic code capacity, we 
note that Ahlswede's Elimination Technique \cite{Ahlswede:78p} 
can be applied to the AVRC as well. Hence, the size of the code collection of any reliable random code can be reduced to polynomial size.

\subsubsection{Deterministic Code Lower and Upper Bounds}
In the next statements, we characterize the deterministic code capacity of the AVRC $\avrc$. 
We consider conditions under which the deterministic code capacity coincides with the random code capacity, and conditions under which it is lower.
For every $x_1\in\Xset_1$, let $\avc_1(x_1)$ and $\avc(x_1)$ denote the marginal AVCs from the sender to the relay and from the sender to the destination receiver, respectively, 
 \begin{align}
\label{eq:RYmarginAVCs}
\avc_1(x_1)=\{ W_{Y_1|X,X_1,S}(\cdot|\cdot,x_1,\cdot) \} \,,\;
\avc(x_1)=\{ W_{Y|X,X_1,S}(\cdot|\cdot,x_1,\cdot) \} \,.
\end{align}
\begin{lemma}
\label{lemm:RYcorrTOdetC} 
If the marginal sender-relay and sender-reciever AVCs have positive capacities, \ie $\opC(\avc_1(x_{1,1}))>0$ and $\opC(\avc(x_{1,2}))$ $>0$, for some $x_{1,1},x_{1,2}\in\Xset_1$, 
 then the capacity of the AVRC $\avrc$ coincides with the random code capacity, \ie 
$\RYCavc = \RYrCav$.  
\end{lemma}
The proof of Lemma~\ref{lemm:RYcorrTOdetC} is given in Appendix~\ref{app:RYcorrTOdetC}.
Next, we give a computable sufficient condition, under which the deterministic code capacity coincides with the random code capacity.
For the point to point AVC, this occurs if and only if the channel is non-symmetrizable \cite{Ericson:85p}\cite[Definition 2]{CsiszarNarayan:88p}. Our condition here is given in terms of an  extended definition of symmetrizability, akin to  
\cite[Definition 9]{HofBross:06p}. 

\begin{definition}
\label{def:RYsymmetrizableGivenX1}
 A state-dependent relay channel $\rc$ is said to be \emph{symmetrizable}-$\Xset|\Xset_1$ if for some conditional distribution $J(s|x)$,
\begin{multline}
\label{eq:symmetrizableGivenX1}
\sum_{s\in\Sset} \rc(y,y_1|x,x_1,s)J(s|\tx)=\sum_{s\in\Sset} \rc(y,y_1|\tx,x_1,s)J(s|x) \,, \\
\forall\, x,\tx\in\Xset \,,\; x_1\in\Xset_1 \,,\;  y\in\Yset \,,\; y_1\in\Yset_1 \,.
\end{multline}
Equivalently, for every given $x_1\in\Xset_1$, the DMC $W_{\bar{Y}|X,X_1,S}(\cdot|\cdot,x_1,\cdot)$ is symmetrizable,
where $\bar{Y}=(Y,Y_1)$. 
\end{definition}

A similar definition applies to the 
marginals $W_{Y|X,X_1,S}$ and $W_{Y_1|X,X_1,S}$. 
\begin{coro}
\label{coro:RYmainDbound}
Let $\avrc$ be an AVRC.
\begin{enumerate}[1)]
\item
If $W_{Y|X,X_1,S}$ and $W_{Y_1|X,X_1,S}$ are non-symmetrizable-$\Xset|\Xset_1$, then  
$\RYCavc = \RYrCav$. In this case, 
\begin{align}
\RYdIRavc\leq \RYCavc \leq \RYrICav \,.
\end{align}
\item
If $\rc$ is reversely degraded, such that $\rc=W_{Y|X,X_1} W_{Y_1|Y,X_1,S}$,
where $W_{Y_1|X,X_1,S}$ is non-\\symmetrizable-$\Xset|\Xset_1$ and
$W_{Y|X,X_1}(y|x,x_1)\neq W_{Y|X,X_1}(y|\tx,x_1)$ for some $x,\tx\in\Xset$, $x_1\in\Xset_1$ and $y\in\Yset$,
 then
\begin{align}
\RYCavc=\RYrCav= \RYdIRavc=\RYrICav=  \min_{q(s)} \max_{p(x,x_1)} I_q(X;Y|X_1) \,.
\end{align}
\item
If $\rc$ is degraded, such that $\rc=W_{Y_1|X,X_1} W_{Y|Y_1,X_1,S}$,
where $W_{Y|X,X_1,S}$ is non-\\symmetrizable-$\Xset|\Xset_1$ and
$W_{Y_1|X,X_1}(y_1|x,x_1)\neq W_{Y_1|X,X_1}(y_1|\tx,x_1)$ for some $x,\tx\in\Xset$, $x_1\in\Xset_1$ and $y_1\in\Yset_1$,
 then
\begin{align}
\RYCavc=\RYrCav= \RYdIRavc=\RYrICav=  \max_{p(x,x_1)}   \min \left\{ \min_{q(s)} I_q(X,X_1;Y) \,,\; 
 I(X;Y_1|X_1)
\right\} \,.
\end{align}
\end{enumerate}
\end{coro}
The proof of Corollary~\ref{coro:RYmainDbound} is given in Appendix~\ref{app:RYmainDbound}.
Note that there are  $4$ symmetrizability cases in terms of the sender-relay channel $W_{Y_1|X,X_1,S}$ and  the sender-receiver channel  $W_{Y|X,X_1,S}$.
For the case where $W_{Y_1|X,X_1,S}$ and $W_{Y|X,X_1,S}$ are both non-symmetrizable-$\Xset|\Xset_1$, the lemma above asserts that the  capacity coincides with the random code capacity. 
In other cases, one may expect the capacity to be lower than the random code capacity. For instance, if $W_{Y|X,X_1,S}$ is non-symmetrizable-$\Xset|\Xset_1$, while $W_{Y_1|X,X_1,S}$ \emph{is} symmetrizable-$\Xset|\Xset_1$, then  the capacity is positive by direct transmission. Furthermore, in this case, if the channel is reversely degraded, then the capacity coincides with the random code capacity. However, 
it remains in question whether this is true in general, when the channel is not reveresly degraded.

Next, we consider conditions under which the capacity is zero.
Observe that if $\rc$ is symmetrizable-$\Xset|\Xset_1$ then so are $W_{Y|X,X_1,S}$ and $W_{Y_1|X,X_1,S}$. Intuitively, this means that 
 the AVRC is a poor channel as well. 
For example, say $Y_1=X+X_1+S$ and $Y=X\cdot X_1\cdot S$, then the jammer can confuse the decoder by taking the state sequence to be some codeword.
The following lemma validates this intuition.
\begin{lemma}
\label{lemm:RYzeroCsymm}
If the AVRC $\avrc$ is symmetrizable-$\Xset|\Xset_1$, then it has zero capacity, \ie $\RYCavc=0$. 
\end{lemma}
Lemma~\ref{lemm:RYzeroCsymm} is proved in Appendix~\ref{app:RYzeroCsymm}.
%
	%
%
If the AVRC is degraded then, we have a simpler symmetrizability condition under which the capacity is zero.
\begin{definition}
\label{def:RYsymmetrizableGivenY1}
Let $\rc=W_{Y_1|X,X_1} W_{Y|Y_1,X_1,S}$ be a degraded relay channel. We say that $\rc$
is \emph{symmetrizable}-$\Xset_1\times\Yset_1$ if for some conditional distribution $J(s|x_1,y_1)$,
\begin{multline}
\label{eq:symmetrizableGivenY1}
\sum_{s\in\Sset} W_{Y|Y_1,X_1,S}(y|y_1,x_1,s)J(s|\tx_1,\ty_1)=\sum_{s\in\Sset} W_{Y|Y_1,X_1,S}(y|\ty_1,\tx_1,s)J(s|x_1,y_1) \,, \\
\forall\,  \tx_1,x_1\in\Xset_1 \,,\;  y\in\Yset \,,\; y_1,\ty_1\in\Yset_1 \,.
\end{multline}
Equivalently, the DMC $W_{Y|\bar{Y}_1,S}$ 
 is symmetrizable, 
 where $\bar{Y}_1=(Y_1,X_1)$. 
\end{definition}
\begin{lemma}
\label{lemm:RYzeroCsymmDeg}
If the AVRC $\avrc$ is degraded and symmetrizable-$\Xset_1\times\Yset_1$, then it has zero capacity, \ie $\RYCavc=0$. 
\end{lemma}
Lemma~\ref{lemm:RYzeroCsymmDeg} is proved in Appendix~\ref{app:RYzeroCsymmDeg}. An example is given 
below. 
\begin{example}
\label{example:BSRCsymm}
Consider a state-dependent relay channel  $\rc$, specified by 
\begin{align}
Y_1=& X+Z \mod 2 \,, \nonumber\\
Y=& X_1+S						 \,, \nonumber
\end{align}
where $\Xset=\Xset_1=\Zset=\Sset=\Yset_1=\{ 0,1 \}$ and $\Yset=\{0,1,2\}$, and the additive noise is distributed according to $Z\sim\text{Bernoulli}(\theta)$,  $0\leq\theta\leq 1$. 
It is readily seen that $\rc$ is degraded and symmetrizable-$\Xset_1\times\Yset_1$, by
(\ref{eq:RYdegraded}) and (\ref{eq:symmetrizableGivenY1}). In particular, (\ref{eq:symmetrizableGivenY1}) is satisfied with $J(s|x_1,y_1)=1$ for 
$s=x_1$, and $J(s|x_1,y_1)=0$ otherwise.
Hence, by Lemma~\ref{lemm:RYzeroCsymmDeg}, the capacity is $\RYCavc=0$.
On the other hand, we show that the random code capacity is given by 
$\RYrCav=\min\left\{\frac{1}{2},1-h(\theta)\right\}$, using Corollary~\ref{coro:RYavrcDeg}.
The derivation is given in Appendix~\ref{app:BSRCsymm}.
\end{example}

\subsection{AVRC with Orthogonal Sender Components}
Consider the special case of  a  relay channel $W_{Y,Y_1|X,X_1,S}$ with orthogonal sender components \cite{ElGamalZahedi:05p} \cite[Section 16.6.2]{ElGamalKim:11b}, where $X=(X',X'')$ and 
\begin{align}
W_{Y,Y_1|X',X'',X_1,S}(y,y_1|x',x'',x_1,s)=
W_{Y|X',X_1,S}(y|x',x_1,s) \cdot W_{Y_1|X'',X_1,S}(y_1|x'',x_1,s)  \,.
\label{eq:OrSrc}
\end{align}
Here, we address the case where the channel output depends on the state only through the relay, \ie
$W_{Y|X',X_1,S}(y|x',x_1,s)=W_{Y|X',X_1}(y|x',x_1)$.
\begin{lemma}
\label{lemm:RYmainOrS}
Let $\avrc$ $=$ $\{W_{Y|X',X_1}$ $W_{Y_1|X'',X_1,S}\}$  be an AVRC with orthogonal sender components. 
The random code capacity  of $\avrc$ is given by 
\begin{align}  
\label{eq:RYICavcOrthoS} 
\RYrCav = \RYdIRavc=\RYrICav=
 \max_{p(x_1) p(x'|x_1) p(x''|x_1)} 
\min \big\{ I(X',X_1;Y) \,,\;
\min_{q(s)} I_{q}(X'';Y_1|X_1)+I(X';Y|X_1) 
\big\} \,.
\end{align}
If $W_{Y_1|X'',X_1,S}$ is non-symmetrizable-$\Xset''|\Xset_1$,  
and  $W_{Y|X',X_1}(y|x',x_1)\neq W_{Y|X',X_1}(y|\tx',x_1)$ for some $x_1\in\Xset_1$, $x',\tx'\in\Xset'$, $y\in\Yset$, 
then the deterministic code capacity is given by  $\RYCavc=\RYdIRavc=\RYrICav$.
\end{lemma}
The proof of Lemma~\ref{lemm:RYmainOrS} is given in
 Appendix~\ref{app:RYmainOrS}.  To prove Lemma~\ref{lemm:RYmainOrS}, we apply the methods of \cite{ElGamalZahedi:05p} to our results. Specifically, we use the partial decode-forward lower bound in Theorem~\ref{theo:RYmain}, taking $U=X''$  (see (\ref{eq:RYdIRcompound}) and (\ref{eq:RYIRcompoundP})).

\begin{appendices}
\section{Proof of Lemma~\ref{lemm:RYcompoundDF}}
\label{app:RYcompoundDF}
\subsection{Partial Decode-Forward Lower Bound}

We construct a block Markov code, where the backward decoder uses joint typicality  with respect to a state type, which is ``close" to some  $q\in\Qset$. Let $\delta>0$ be arbitrarily small.
%
Define a set of state types $\tQ$ by  
\begin{align}
\label{eq:RYtQ}
\tQ=\{ \hP_{s^n} \,:\; s^n\in\Aset^{ \delta_1 
}(q) \;\text{ for some  $q\in\Qset
$}\, \} \,,
\end{align}
where 
\begin{align}
\label{eq:RYdelta1}
\delta_1 \triangleq
\frac{\delta}{2\cdot |\Sset|} \,.
\end{align} 
Namely, $\tQ$ is the set of types that are $\delta_1$-close 
 to some state distribution $q(s)$ in $\Qset$. 
A code $\code$ for the compound relay channel 
 is constructed as follows.

The encoders use $B$ blocks, each consists of $n$ channel uses to convey $(B-1)$ independent messages to the receiver. Furthermore, each message $M_b$, for $b\in [1:B-1]$,  is divided into two independent messages.
That is, $M_b=(M_b',M_b'')$,  where $M_b'$ and $M_b''$ are uniformly distributed, \ie
\begin{align}
M_b'\sim \text{Unif}[1:2^{nR'}] \,,\; M_b''\sim \text{Unif}[1:2^{nR''}] \,,\;\text{with $R'+R''=R$}\,,
\end{align}
for $b\in [1:B-1]$.
For convenience of notation, set $M_0'=M_B'\equiv 1$ and $M_0''=M_B''\equiv 1$. The average rate $\frac{B-1}{B}\cdot R$ is arbitrarily close to $R$.

\emph{Codebook Generation}: 
 Fix the distribution $P_{U,X,X_1}(u,x,x_1)$, and let 
\begin{align}
\label{eq:RYCompoundDistUY}
P^q_{X,Y,Y_1|U,X_1}(x,y,y_1|u,x_1)=P_{X|U,X_1}(x|u,x_1) \sum_{s\in\Sset} q(s)  \rc(y,y_1|x,x_1,s) \,.
\end{align}
We construct $B$ independent codebooks. 
For $b\in [2:B-1]$, generate $2^{n R'}$ independent sequences $x_{1,b}^n(m_{b-1}')$, $m_{b-1}'\in [1:2^{nR'}]$, at random, each according to $\prod_{i=1}^n P_{X_1}(x_{1,i})$.
Then, generate $2^{nR'}$ sequences,
\begin{align}
u^n_b(m_b'|m_{b-1}') \sim& \prod_{i=1}^n P_{U|X_1}(u_i|x_{1,b,i}(m_{b-1}')) \,,\; m_b'\in [1:2^{nR'}] \,,
\end{align}
 conditionally independent given $x_{1,b}^n(m_{b-1}')$. 
Then, for every $m_b'\in [1:2^{nR'}]$, generate $2^{nR''}$ sequences,
\begin{align}
x_b^n(m_b',m_b''|m_{b-1}') \sim& \prod_{i=1}^n P_{X|U,X_1}(x_i|u_{b,i}(m_b'|m_{b-1}'),x_{1,b,i}(m_{b-1}')) \,,\; m_b''\in [1:2^{nR''}] \,,
\end{align}
 conditionally independent given $(u_b^n(m_b'|m_{b-1}'),x_{1,b}^n(m_{b-1}'))$. 
We have thus generated $B-2$ independent codebooks,
\begin{align}
\Fset_b=\Big\{ \left( x_{1,b}^n(m_{b-1}'), u_b^n(m_b'|m_{b-1}'), x_b^n(m_b',m_b''|m_{b-1}')    \right) \,:\;
 m_{b-1}',m_b'\in [1:2^{nR'}] \,,\; m_b''\in [1:2^{nR''}]
\Big\} \,,
\end{align}
for $b\in [2:B-1]$. The codebooks $\Fset_1$ and $\Fset_B$ are generated in the same manner, with fixed $m_0'=m_B'\equiv 1$ and $m_0''=m_B''\equiv 1$. 
Encoding and decoding is illustrated in Figure~\ref{fig:RYpDFcompound}.

\emph{Encoding}:
To send the message sequence $(m_1',m_1'',\ldots, m_{B-1}',m_{B-1}'')$, transmit $x_b^n(m_b',m_b''|m_{b-1}')$ at block $b$, for 
$b\in [1:B]$. 

\emph{Relay Encoding}:
In block $1$, the relay transmits $x_{1,1}^n(1)$. Set $\tm_0'\equiv 1$.
 At the end of block $b\in [1:B-1]$, the relay receives $y_{1,b}^n$, and
 finds some $\tm_b'\in [1:2^{nR'}]$ such that 
\begin{align}
(u_b^n(\tm_b'|\tm_{b-1}'),x_{1,b}^n(\tm_{b-1}'),y_{1,b}^n)\in\tset(P_{U,X_1} P^{q}_{Y_1|U,X_1}) \,,\; 
\text{for some $q\in\tQ$}\,.
\end{align}
If there is none or there is more than one such, set $\tm_b'=1$. 
 In block $b+1$, the relay transmits $x_{1,b+1}^n(\tm_b')$.

\emph{Backward Decoding}: Once all blocks $(y_b^n)_{b=1}^B$ are received, decoding is performed backwards.
Set $\hm_B'=\hm_B''\equiv 1$. 
For $b=B-1,B-2,\ldots,1$, find 
 a unique $\hm_b'\in[1:2^{nR'}]$ such that
\begin{align}
 (u_{b+1}^n(\hm_{b+1}'|\hm_b'), x_{1,b+1}^n(\hm_b'),y_{b+1}^n)\in\tset(P_{U,X_1} P^{q}_{Y|U,X_1}) \,,\;\text{for some $q\in\tQ$}\,.
\end{align}
If there is none, or more than one such $\hm_b'\in[1:2^{nR'}]$, declare an error.

Then, the decoder uses $\hm_1',\ldots,\hm_{B-1}'$ as follows.
For $b=B-1,B-2,\ldots,1$, find 
 a unique $\hm_b''\in[1:2^{nR''}]$ such that
\begin{align}
 (u_{b}^n(\hm_{b}'|\hm_{b-1}'),x_{b}^n(\hm_{b}',\hm_{b}''|\hm_{b-1}'), x_{1,b}(\hm_{b-1}'),y_{b}^n)\in\tset(P_{U,X,X_1} P^{q}_{Y|X,X_1}) \,,\; \text{for some $q\in\tQ$}\,.
\end{align}  
If there is none, or more than one such $\hm_b''\in[1:2^{nR''}]$, declare an error.
We note that using the set of types $\tQ$ instead of the original set of state distributions $\Qset$ alleviates the analysis, since
 $\Qset$ is not necessarily finite nor countable.

\begin{center}
\begin{figure}

\begin{tabular}{l|ccccc}
Block				& $1$							& $2$									
									& $\cdots$			& $B-1$ 		& $B$\\
								
\\ \hline &&&&& \\ 
Encoder				&	$x_{1}^n(m_1',m_1''|1)$	&	$x_{2}^n(m_2',m_2''|m_1')$		
							& $\cdots$			& $x_{B-1}^n(m_{B-1}',m_{B-1}''|m_{B-2}')$&  $x_{B}^n(1,1|m_{B-1}')$ 
\\&&&&& \\
Relay Decoder			& $\tm_1'\rightarrow$				& $\tm_2'\rightarrow$						
						& $\cdots$			& $\tm_{B-1}'$ & $\emptyset$
\\&&&&& \\
Relay Encoder				&	$x_{1,1}^n(1)$	&	$x_{1,2}^n(\tm_1')$		
							& $\cdots$			& $x_{1,B-1}^n(\tm_{B-2}')$&  $x_{1,B}^n(m_{B-1}')$ \\&&&&& \\
Output 		& $\emptyset$		& $\hm_1'$ 
					& $\cdots$ & $\leftarrow\hm_{B-2}'$ & $\leftarrow\hm_{B-1}'$ 
					\\
					& $\hm_1''$		& $\hm_2''$ 
					& $\cdots$ & $\hm_{B-1}''$ & $\emptyset$
\end{tabular}
\caption{Partial decode-forward coding scheme. The block index $b\in [1:B]$ is indicated at the top. In the following rows, we have the corresponding elements: 
(1) sequences transmitted by the encoder;  
(2) estimated messages at the relay;
(3) sequences transmitted by the relay;
(4) estimated messages at the destination decoder.
The arrows in the second row indicate that the relay encodes forwards with respect to the block index, while the arrows in the fourth row indicate that the receiver decodes backwards.  
}
\label{fig:RYpDFcompound}
\end{figure}
\end{center}

\emph{Analysis of Probability of Error}: 
Assume without loss of generality that the user sent  $(M_b',M_b'')=(1,1)$, and let $q^*(s)\in\Qset$ denote the \emph{actual} state distribution chosen by the jammer.
The error event is bounded by the union of the  events
\begin{align}
\Eset_{1}(b)=&\{ \tM_b'\neq 1 \} \,,\; \Eset_{2}(b)=\{ \hM_b'\neq 1 \} \,,\; \Eset_{3}(b)=\{ \hM_b''\neq 1 \} \,,\; \text{for $b\in [1:B-1]$}\,.
\end{align}
 Then,  the probability of error is bounded by
\begin{align}
 \err(q,\code) 
\leq& \sum_{b=1}^{B-1} \prob{\Eset_1(b)}+ \sum_{b=1}^{B-1} \cprob{\Eset_2(b)}{\Eset_1^c(b) }
+ \sum_{b=1}^{B-1} \cprob{\Eset_3(b)}{\Eset_1^c(b)\cap \Eset_2^c(b) \cap \Eset_2^c(b-1) }  \,,
\label{eq:RYCompoundCerrBound}
\end{align}
with $\Eset_2(0)=\emptyset$, where the conditioning on $(M'_b,M_b'')=(1,1)$ is omitted for convenience of notation.

We begin with the probability of erroneous relaying, $\prob{\Eset_{1}(b)}$. Define
\begin{align}
\Eset_{1,1}(b)=& \{ (  U_b^n(1|\tM_{b-1}'), X_{1,b}^n(\tM_{b-1}'), Y_{1,b}^n  )\notin \tset(P_{U,X_1} P^{q'}_{Y_1|U,X_1})  \;\text{ for all $q'\in\tQ$} \} \nonumber\\
\Eset_{1,2}(b)=& \{ (  U_b^n(m_b'|\tM_{b-1}'), X_{1,b}^n(\tM_{b-1}'), Y_{1,b}^n  )\in \tset(P_{U,X_1} P^{q'}_{Y_1|U,X_1}) \,,\; 
\text{for some $m_b'\neq 1$, $q'\in\tQ$} \} \,.
\label{eq:RYcompuondE11}
\end{align}
For $b\in [1:B-1]$, the relay error event is bounded as 
\begin{align}
\Eset_1(b)\subseteq& \Eset_1(b-1)\cup \Eset_{1,1}(b)\cup \Eset_{1,2}(b) \nonumber\\
=& \Eset_1(b-1)\cup \left( \Eset_1(b-1)^c \cap \Eset_{1,1}(b) \right) \cup
\left( \Eset_1(b-1)^c \cap \Eset_{1,2}(b) \right) \,,
\end{align}
with $\Eset_1(0)=\emptyset$. 
 Thus, by the union of events bound,
\begin{align}
\label{eq:RYcompoundDFE1bound}
\prob{\Eset_1(b)}\leq \prob{\Eset_1(b-1)}+\cprob{ \Eset_{1,1}(b)}{\Eset_1(b-1)^c}
+\cprob{ \Eset_{1,2}(b)}{\Eset_1(b-1)^c} \,.
\end{align}
Consider the second term on the RHS of (\ref{eq:RYcompoundDFE1bound}).
We now claim that given that $\Eset_1(b-1)^c$ occurred, \ie $\tM_{b-1}'=1$,   the event $\Eset_{1,1}(b)$ implies that 
$(  U_b^n(1|1), X_{1,b}^n(1), Y_{1,b}^n  )\notin \Aset^{\nicefrac{\delta}{2}}(P_{U,X_1} P^{q''}_{Y_1|U,X_1})$
	 for all $q''\in\Qset$.	This claim is due to the following. 
Assume to the contrary that $\Eset_{1,1}(b)$ holds, but $(  U_b^n(1|1), X_{1,b}^n(1), Y_{1,b}^n  )\in \Aset^{\nicefrac{\delta}{2}}(P_{U,X_1} P^{q''}_{Y_1|U,X_1})$ for some $q''\in\Qset$.  
Then, for a sufficiently large $n$, there exists a type $q'(s)$ such that 
\begin{align}
\label{eq:RYqdelta1}
|q'(s)-q''(s)|\leq \delta_1 \,, 
\end{align}
 for all $s\in\Sset$, and by the definition in (\ref{eq:RYtQ}), $q'\in\tQ$.  Then, (\ref{eq:RYqdelta1}) implies that
\begin{align}
|P_{Y_1|U,X_1}^{q'}(y_1|u,x_1)-P_{Y_1|U,X_1}^{q''}(y_1|u,x_1)|\leq |\Sset|\cdot \delta_1=\frac{\delta}{2} \,,
\end{align}
for all $u\in\Uset$, $x_1\in\Xset_1$ and $y_1\in\Yset_1$ (see (\ref{eq:RYCompoundDistUY}) and (\ref{eq:RYdelta1})). Hence, 
$(  U_b^n(1|1), X_{1,b}^n(1), Y_{1,b}^n  )\in \tset(P_{U,X_1} P^{q'}_{Y_1|U,X_1})$, which contradicts the first assumption.
It follows that 
	\begin{align}
	\cprob{ \Eset_{1,1}(b)}{\Eset_1(b-1)^c} 
	\leq& \cprob{(  U_b^n(1|1), X_{1,b}^n(1), Y_{1,b}^n  )\notin \Aset^{\nicefrac{\delta}{2}}(P_{U,X_1} P^{q''}_{Y_1|U,X_1}) \;\text{ for all $q''\in\Qset$} }{\Eset_1(b-1)^c} \nonumber \\
	\leq& \cprob{(  U_b^n(1|1), X_{1,b}^n(1), Y_{1,b}^n  )\notin \Aset^{\nicefrac{\delta}{2}}(P_{U,X_1} P^{q^*}_{Y_1|U,X_1})  }{\Eset_1(b-1)^c} \,.
	\label{eq:RYRYSllnRL}
	\end{align}
		Since the codebooks $\Fset_1,\ldots,\Fset_B$ are independent, the sequence $(U_b^n(1|1), X_{1,b}^n(1))$ from the codebook $\Fset_b$ is independent of the relay estimate $\tM_{b-1}$, which is a function of $Y_{1,b-1}^n$ and the codebook $\Fset_{b-1}$. Thus, 
the RHS of (\ref{eq:RYRYSllnRL})  tends to zero exponentially as $n\rightarrow\infty$ by the law of large numbers and  Chernoff's bound.

	We move to the third term in the RHS of (\ref{eq:RYcompoundDFE1bound}). 
	By the union of events bound, the fact that the number of type classes in $\Sset^n$ is bounded by $(n+1)^{|\Sset|}$, and the independence of the codebooks,  we have that  
\begin{align}
\label{eq:RYSE2poly}
&\cprob{ \Eset_{1,2}(b)}{\Eset_1(b-1)^c} 
\leq (n+1)^{|\Sset|}\cdot \sup_{q'\in\tQ} \prob{
(  U_b^n(m_b'|1), X_{1,b}^n(1), Y_{1,b}^n  )\in \tset(P_{U,X_1} P^{q'}_{Y_1|U,X_1}) \;\text{ for some $m_b'\neq 1$} 
} \nonumber\\
\leq& (n+1)^{|\Sset|}\cdot 2^{nR'} \cdot
 \sup_{q'\in\tQ}\left[  \sum_{u^n,x_1^n} P_{U^n,X_1^n}(u^n,x_1^n) \cdot \sum_{y_1^n \,:\; (u^n,x_1^n,y_1^n)\in \tset(P_{U,X_1} P^{q'}_{Y_1|U,X_1})} P_{Y_1^n|X_1^n}^{q^*}(y_1^n|x_1^n)
\right] \,,
\end{align}
where the last line follows since $U_b^n(m_b'|1)$ is conditionally independent of $Y_{1,b}^n$ given $X_{1,b}^n(1)$, for every $m_b'\neq 1$. 
 Let $y_1^n$ satisfy $(u^n,x_1^n,y_1^n)\in \tset(P_{U,X_1} P^{q'}_{Y_1|U,X_1})$. Then, $\,(x_1^n,y_1^n)\in\Aset^{\delta_2}(P_{X_1,Y_1}^{q'})$ with $\delta_2\triangleq |\Uset| \cdot\delta$. By Lemmas 2.6 and 2.7 in \cite{CsiszarKorner:82b},
\begin{align*}
P_{X_1^n,Y_1^n}^{q^*}(x_1^n,y_1^n)=2^{-n\left(  H(\hP_{x_1^n,y_1^n})+D(\hP_{x_1^n,y_1^n}||P_{X_1,Y_1}^{q^*})
\right)}\leq 2^{-n H(\hP_{x_1^n,y_1^n})}
\leq 2^{-n\left( H_{q'}(X_1,Y_1) -\eps_1(\delta) \right)} \,,
\end{align*}
hence,
\begin{align}
\label{eq:RYSpYbound}
P_{Y_1^n|X_1^n}^{q^*}(y_1^n|x_1^n)\leq 2^{-n\left( H_{q'}(Y_1|X_1) -\eps_2(\delta) \right)} \,,
\end{align}
where $\eps_1(\delta),\eps_2(\delta)\rightarrow 0$ as $\delta\rightarrow 0$. Therefore, by (\ref{eq:RYSE2poly})$-$(\ref{eq:RYSpYbound}), along with 
 \cite[Lemma 2.13]{CsiszarKorner:82b},
\begin{align}
& \cprob{ \Eset_{1,2}(b)}{\Eset_1(b-1)^c}         																									
\leq
 \;(n+1)^{|\Sset|}\cdot \sup_{q'\in\Qset} 
2^{-n[ I_{q'}(U;Y_1|X_1) 
-R'-\eps_3(\delta) ]} \label{eq:RYSLexpCR} \,,
\end{align}
with $\eps_3(\delta)\rightarrow 0$ as $\delta\rightarrow 0$.
 Using induction, we have by (\ref{eq:RYcompoundDFE1bound}) that $\prob{\Eset_1(b)}$  tends to zero exponentially as $n\rightarrow\infty$, for $b\in [1:B-1]$, provided that $R'<\inf_{q'\in\Qset} I_{q'}(U;Y_1|X_1)
-\eps_3(\delta)$.

As for the erroneous decoding of $M_b'$ at the receiver, 
observe that given $\Eset_1(b)^c$, the relay sends $X_{1,b}^n(1)$ in block $b+1$, hence
\begin{align} 
( U_{b+1}^n(1|1), X_{b+1}^n(1,1|1), X_{1,b+1}^n(1)
)\sim P_{U,X,X_1}(u,x,x_1) \,. 
\label{eq:RYcompoundDFtrueDist}
\end{align}
At the destination receiver,  decoding is performed backwards, hence the error events 
 have a different form compared to those of the relay (\cf (\ref{eq:RYcompuondE11}) and the events below). 
Define the events,
\begin{align}
\Eset_{2,1}(b)=& \{ (  U_{b+1}^n(\hM_{b+1}'|1), X_{1,b+1}^n(1), Y_{b+1}^n  )\notin \tset(P_{U,X_1} P^{q'}_{Y|U,X_1})  \;\text{ for all $q'\in\tQ$} \} \nonumber\\
\Eset_{2,2}(b)=& \{ (  U_{b+1}^n(\hM_{b+1}'|m_b'), X_{1,b+1}^n(m_b'), Y_{b+1}^n  )\in \tset(P_{U,X_1} P^{q'}_{Y_1|U,X_1}) \,,\; 
\text{for some $m_b'\neq 1$, $q'\in\tQ$} \} 
\end{align}
For $b\in [1:B-1]$, the error event $\Eset_2(b)$ is bounded by 
\begin{align}
\Eset_2(b)\subseteq& \Eset_2(b+1)\cup \Eset_{2,1}(b)\cup \Eset_{2,2}(b) \nonumber\\
=& \Eset_2(b+1)\cup \left( \Eset_2(b+1)^c \cap \Eset_{2,1}(b) \right) \cup
\left( \Eset_2(b+1)^c \cap \Eset_{2,2}(b) \right) \,,
\end{align}
with $\Eset_2(B)=\emptyset$. Thus,
\begin{align}
\cprob{\Eset_2(b)}{\Eset_1(b)^c}\leq& 
\cprob{\Eset_2(b+1)}{\Eset_1(b)^c}+\cprob{\Eset_{2,1}(b)}{\Eset_1(b)^c,\Eset_2(b+1)^c}
 +\cprob{\Eset_{2,2}(b)}{\Eset_1(b)^c,\Eset_2(b+1)^c} \,.
\label{eq:RYcompoundE2b}
\end{align}
By similar arguments to those used above, we have that 
\begin{align}
\cprob{\Eset_{2,1}(b)}{\Eset_1(b)^c,\Eset_2(b+1)^c}\leq
\cprob{(  U_{b+1}^n(1|1), X_{1,b+1}^n(1), Y_{b+1}^n  )\notin \Aset^{\nicefrac{\delta}{2}}(P_{U,X_1} P^{q^*}_{Y|U,X_1}) }{\Eset_1(b)^c} \,,
\label{eq:RYcompoundE21b}
\end{align}
which  tends to zero exponentially as $n\rightarrow\infty$, due to (\ref{eq:RYcompoundDFtrueDist}), and
 by the law of large numbers and Chernoff's bound. Then, by similar arguments to those used for the bound on $\cprob{ \Eset_{1,2}(b)}{\Eset_1(b-1)^c} $, the third term on the RHS of (\ref{eq:RYcompoundE2b}) tends to zero as $n\rightarrow\infty$, provided that 
$R'<\inf_{q'\in\Qset} I_{q'}(U,X_1;Y)-\eps_4(\delta)$, where $\eps_4(\delta)\rightarrow 0$ as $\delta\rightarrow 0$.  Using induction, we have by (\ref{eq:RYcompoundE2b}) that the second term on the RHS of (\ref{eq:RYCompoundCerrBound})
 tends to zero exponentially as $n\rightarrow\infty$, for $b\in [1:B-1]$. 

Moving to the error event for $M_b''$, define
\begin{align}
\Eset_{3,1}(b)=& \{ (  U_b^n(\hM_b'|\hM_{b-1}'), X_{b}^n(\hM_{b}',1|\hM_{b-1}'), X_{1,b}(\hM_{b-1}'),
 Y_{b}^n  )\notin \tset(P_{U,X,X_1} P^{q'}_{Y|X,X_1}) \,,\; 
\text{for all $q'\in\tQ$} \} \nonumber\\
\Eset_{3,2}(b)=& \{ (  U_b^n(\hM_b'|\hM_{b-1}'), X_{b}^n(\hM_{b}',m_b''|\hM_{b-1}'), X_{1,b}(\hM_{b-1}'), Y_{b}^n  )\in \tset(P_{U,X,X_1} P^{q'}_{Y|X,X_1}) \,,\; 
\text{for some $m_b''\neq 1$, $q'\in\tQ$} \} \,.
\end{align}
Given $\Eset_2(b)^c\cap \Eset_2(b-1)^c$, we have that $\hM_b'=1$ and $\hM_{b-1}'=1$. 
Then, by similar arguments to those used above,
\begin{align}
&\cprob{\Eset_3(b)}{\Eset_1(b)^c\cap \Eset_2(b)^c\cap \Eset_2(b-1)^c} \nonumber\\
\leq&  \cprob{\Eset_{3,1}(b)}{\Eset_1(b)^c\cap \Eset_2(b)^c\cap \Eset_2(b-1)^c}
 +\cprob{\Eset_{3,2}(b)}{\Eset_1(b)^c\cap \Eset_2(b)^c\cap \Eset_2(b-1)^c} \nonumber\\
\leq& e^{-a_0 n}
+(n+1)^{|\Sset|}\cdot \sup_{q'\in\Qset} \sum_{m_b''\neq 1}\cprob{(  U_b^n(1|1), X_{b}^n(1,m_b''|1), X_{1,b}(1), Y_{b}^n  )\in \tset(P_{U,X,X_1} P^{q'}_{Y|X,X_1})}{\Eset_1(b)^c}
\nonumber\\
\leq&e^{-a_0 n}+  (n+1)^{|\Sset|}\cdot \sup_{q'\in\Qset} 
2^{-n[ I_{q'}(X;Y|U,X_1) -R''-\eps_5(\delta) ]}
\label{eq:RYcompoundE3b}
\end{align}
where $a_0>0$ and $\eps_5(\delta)\rightarrow 0$ as $\delta\rightarrow 0$. 
The second inequality holds by (\ref{eq:RYcompoundDFtrueDist}) along with the law of large numbers and Chernoff's bound, and the last inequality holds as $X^n_b(1,m_b''|1)$ is conditionally independent of 
$Y_b^n$ given $(U_b^n(1|1),X_{1,b}^n(1))$ for every $m_b''\neq 1$.
Thus, the third term on the RHS of (\ref{eq:RYCompoundCerrBound})  tends to zero exponentially as $n\rightarrow\infty$, provided that $R''< \inf_{q'\in\Qset} I_{q'}(X;Y|U,X_1)-\eps_5(\delta)$.
Eliminating $R'$ and $R''$, we conclude that the probability of error, averaged over the class of the codebooks, exponentially decays to zero  as $n\rightarrow\infty$, provided that $R<\RYdIRcompound$. Therefore, there must exist a $(2^{nR},n,\eps)$ deterministic code, for a sufficiently large $n$.
\qed

\subsection{Cutset Upper Bound}
This is a straightforward consequence of the cutset bound in \cite{CoverElGamal:79p}.
Assume to the contrary that there exists an achievable rate $R>\RYICcompound$. Then, for some 
$q^*(s)$ in the closure of $\Qset$, 
\begin{align}
\label{eq:cutsetCompound1}
R>\max_{p(x,x_1)} 
\min \left\{ I_{q^*}(X,X_1;Y) \,,\; I_{q^*}(X;Y,Y_1|X_1)
\right\} \,.
\end{align}
By the achievability assumption, we have that for every $\eps>0$ and sufficiently large $n$, there exists a $(2^{nR},n)$ random code $\code^\Gamma$ such that $\err(q,\code)\leq\eps$ for every i.i.d. state distribution
 $q\in\Qset$, and in particular for $q^*$. This holds even if $q^*$ is in the closure of $\Qset$ but not in $\Qset$ itself, since $\err(q,\code)$ is continuous in $q$.
Consider using this code over a standard relay channel $W_{Y,Y_1|X,X_1}$ without a state, where
$
W_{Y,Y_1|X,X_1}(y,y_1|x,x_1)=\sum_{s\in\Sset} q^*(s) \rc(y,y_1|x,x_1,s) 
$. 
It follows that the rate $R$ as in (\ref{eq:cutsetCompound1}) can be achieved over the relay channel $W_{Y,Y_1|X,X_1}$, in contradiction to \cite{CoverElGamal:79p}. We deduce that the assumption is false, and $R>\RYICcompound$ cannot be achieved.
\qed

\section{Proof of Corollary~\ref{coro:RYcompoundDeg}}
\label{app:RYcompoundDeg}
This is a straightforward consequence of Lemma~\ref{lemm:RYcompoundDF}, which states that
the capacity of the compound relay channel is bounded by 
$\RYdIRcompound\leq\RYCcompound\leq \RYICcompound$.
Thus, if $\rc$ is reversely degraded such that $\rc=W_{Y|X,X_1} W_{Y_1|Y,X_1,S}$, then  $I_q(X;Y,Y_1|X_1)=I_q(X;Y|X_1)$, and
the bounds coincide by the minimax theorem \cite{sion:58p}, \cf (\ref{eq:RYICcompound}) and (\ref{eq:RYcompoundDirectTran}).
Similarly, if $\rc$ is degraded such that $\rc=W_{Y_1|X,X_1} W_{Y|Y_1,X_1,S}$, then  $I_q(X;Y,Y_1|X_1)=I(X;Y_1|X_1)$, and by (\ref{eq:RYICcompound}) and (\ref{eq:RYcompoundFullDF}), 
\begin{align}
\RYICcompound=& \min_{q(s)\in\Qset}  \max_{p(x,x_1)} 
\min \left\{ I_q(X,X_1;Y) \,,\; I(X;Y_1|X_1)
\right\} \ \,, 
\label{eq:RYdegradedICcompound}
\\
\RYdIRcompound=&   \max_{p(x,x_1)} \min_{q(s)\in\Qset} 
\min \left\{ I_q(X,X_1;Y) \,,\; I(X;Y_1|X_1)
\right\} \,.
\label{eq:RYdegradedIRcompound}
\end{align} 
Observe that $\min \left\{ I_q(X,X_1;Y) \,,\; I(X;Y_1|X_1)
\right\}$ is concave in $p(x,x_1)$ and quasi-convex in $q(s)$ (see \eg \cite[Section 3.4]{BoydVandenbergh:04b}), hence  the bounds (\ref{eq:RYdegradedICcompound}) and (\ref{eq:RYdegradedIRcompound}) coincide by the minimax theorem \cite{sion:58p}.
\qed

\section{Proof of Corollary~\ref{coro:RYcompoundDFb}}
\label{app:RYcompoundDFb}
Consider the block-compound relay channel $\avrc^{\Qset\times B}$, where the state distribution $q_b\in\Qset$
varies from block to block. Since the encoder, relay and receiver are aware of this jamming scheme, the capacity is the same as that of the ordinary compound channel, \ie $\opC(\avrc^{\Qset\times B})=\RYCcompound$ and
$\opC^{\rstarC}(\avrc^{\Qset\times B})=\RYrCcompound$. Hence, 
(\ref{eq:RYcDFb}) and (\ref{eq:RYcCSb}) follow from Lemma~\ref{lemm:RYcompoundDF}. As for the second part of the corollary, observe that the block Markov coding scheme used in the proof of the decode forward lower bound  can be applied as is to the block-compound relay channel, since the relay and the destination receiver do not estimate the state distribution while decoding the messages (see Appendix~\ref{app:RYcompoundDF}). Furthermore, the analysis also holds, where the actual state distribution $q^*$, in (\ref{eq:RYRYSllnRL})--(\ref{eq:RYSpYbound}) and (\ref{eq:RYcompoundE21b}),  is now replaced  by the state distribution $q^*_b$ which corresponds to block $b\in [1:B]$.
\qed

\section{Proof of Theorem~\ref{theo:RYmain}}
\label{app:RYmain}
First, we explain the general idea. 
We modify Ahlswede's Robustification Technique (RT)  \cite{Ahlswede:86p} to the relay channel. Namely, we use codes for the compound  relay channel to construct a random code for the AVRC using randomized permutations. However, in our case, the strictly causal nature of the relay imposes a difficulty, and the application of the RT is not straightforward.

In \cite{Ahlswede:86p}, there is noncausal state information  and a random code is defined via permutations of the codeword symbols and the received sequence. Here, however, the relay cannot apply permutations to  its transmission $x_1^n$, because it depends on the received sequence $y_1^n$ in a strictly causal manner. 
We resolve this difficulty using block Markov codes for the block-compound  relay channel to construct a random code for the AVRC, applying $B$ 
 in-block permutations to the relay transmission, which depends only on the sequence received in the \emph{previous block}. The details are given below.

\subsection{Partial Decode Forward Lower Bound}
We show that every rate $R<\RYdIRavc$  (see (\ref{eq:RYIRcompoundP})) can be achieved by random codes over the AVRC $\avrc$, 
\ie $\RYCavc \geq \RYdIRavc$.
We start with Ahlswede's RT \cite{Ahlswede:86p}, stated below. Let $h:\Sset^n\rightarrow [0,1]$ be a given function. If, for some fixed $\alpha_n\in(0,1)$, and for all 
$ \qn(s^n)=\prod_{i=1}^n q(s_i)$, with 
$q\in\pSpace(\Sset)$, 
\begin{align}
\label{eq:RYRTcondC}
\sum_{s^n\in\Sset^n} \qn(s^n)h(s^n)\leq \alpha_n \,,
\end{align}
then,
\begin{align}
\label{eq:RYRTresC}
\frac{1}{n!} \sum_{\pi\in\Pi_n} h(\pi s^n)\leq \beta_n \,,\quad\text{for all $s^n\in\Sset^n$} \,,
\end{align}
where $\Pi_n$ is the set of all $n$-tuple permutations $\pi:\Sset^n\rightarrow\Sset^n$, and 
$\beta_n=(n+1)^{|\Sset|}\cdot\alpha_n$. 

According to Corollary~\ref{coro:RYcompoundDFb}, 
 for every $R<\RYdIRavc$, there exists a  $(2^{nR(B-1)},$  $nB,$ $e^{-2\theta n})$ block Markov code for the block-compound relay channel $\avrc^{\pSpace(\Sset)\times B}$ 
for some $\theta>0$ and sufficiently large $n$, where $B>0$ is arbitrarily large. 
Recall that the code constructed in the proof in Appendix~\ref{app:RYcompoundDF} has the following form.
The encoders use $B>0$ blocks to convey $B-1$ messages $m_b$, $b\in [1:B-1]$. Each message consists of two parts, \ie $m_b=(m_b',m_b'')$, where $m_b'\in[1:2^{nR'}]$ and $m_b''\in[1:2^{nR''}]$.
In block $b\in [1:B]$, the encoder sends $x_b^n=f_b(m_b',m_b''|m_{b-1}')$, with fixed $m_0$ and $m_B$, 
 and the relay transmits $x_{1,b}^n=f_{1,b}(y_{1,b-1}^n)$, using the sequence received in the previous block.
After receiving the entire output sequence $(y_b^n)_{b=1}^B$, 
the decoder finds an estimate for the messages.
Set $\hm_B'=1$. The first part of each message is decoded backwards as 
 $\hm_b'=g_b'(y_{b+1}^n,\hm_{b+1}')$, for $b=B-1,B-2,\ldots,1$. 
Then, the second part of each message is decoded as $\hm_b''=g_b''(y_b^n,\hm_1',\ldots,\hm_{B-1}')$,
for $b\in [1:B-1]$.
The overall blocklength is then $n\cdot B$ and the average rate is $\frac{B-1}{B} (R'+R'')$.

Given such a block Markov code $\code_{BM}$ for the block-compound relay channel $\avrc^{\pSpace(\Sset)\times B}$, 
 we have that 
\begin{align}
&\text{Pr}_{\,\code_{BM}}\hspace{-0.1cm} \left(\Eset_b' \,|\; (\Eset_{b+1}')^c \right) 
 \leq e^{-2\theta n} \,,\;
\text{Pr}_{\,\code_{BM}}\hspace{-0.1cm} \left(\Eset_b''\,|\;  \Eset_1'^c,\ldots,\Eset_{b-1}'^c \right) 
 \leq e^{-2\theta n}
\end{align}
for $b=B-1,\ldots,1$, where $\Eset_{0}'=\Eset_{B}'=\emptyset$, and
$\Eset_b'=\{ \hM_b'\neq M_b' \}$,  
$\Eset_b''=\{ \hM_b''\neq M_b'' \}$, 
$b\in [1:B-1]$. 
That is, for every sequence of state distributions $q_1,\ldots,q_{b+1}$, where  $ \qn_t(s_t^n)=\prod_{i=1}^n q_t(s_{t,i})$ for
$t\in [1:b+1]$,
\begin{align}
&\sum_{s_{1}^n\in\Sset^n} \qn_{1}(s_{1}^n)\sum_{s_{2}^n\in\Sset^n} \qn_{2}(s_{2}^n)\
\cdots\sum_{s_{b+1}^n\in\Sset^n} \qn_{b+1}(s_{b+1}^n)\cdot h_b'(s_{1}^n,s_{2}^n,\ldots, s_{b+1}^n) \leq e^{-2\theta n} \,,\;
\label{eq:RYrCh1eps}
\intertext{and}\;
&\sum_{s_{1}^n\in\Sset^n} \qn_{1}(s_{1}^n)\sum_{s_{2}^n\in\Sset^n} \qn_{2}(s_{2}^n)\
\cdots\sum_{s_{b}^n\in\Sset^n} \qn_{b}(s_{b}^n) \cdot h_b''(s_{1}^n,s_{2}^n,\ldots,s_{b}^n) \leq e^{-2\theta n} \,,
\label{eq:RYrCh2eps}
\end{align}
where
\begin{align}
&h_b'(s_{1}^n,s_{2}^n,\ldots, s_{b+1}^n)=  
 \frac{1}{2^{n(b+1)(R'+R'')}}\sum_{(m_1',m_1''),\ldots,(m_{b+1}',m_{b+1}'')}  \nonumber\\&
\sum_{y_{1,b}^n\in\Yset_1^n} \cprob{Y_{1,b}^n=y_{1,b}^n}{
(M_1',M_1'')=(m_1',m_1''),\ldots,(M_b',M_b'')=(m_b',m_b''),S_1^n=s_{1}^n,\ldots, S_b^n= s_{b}^n} \nonumber\\ &\;\times
\sum_{y_{b+1}^n : g_b'(y_{b+1}^n,m_{b+1}')\neq m_b' } W_{Y^n|X^n,X_{1}^n,S^n}(y_{b+1}^n| 
f_{b+1}(m_{b+1}',m_{b+1}''|m_b'),f_{1,b+1}(y_{1,b}^n),s_{b+1}^n  )
\label{eq:RYrch1}
\intertext{and}
&h_b''(s_{1}^n,s_{2}^n,\ldots,s_{b}^n) =
 \frac{1}{2^{nR''}}\sum_{m_{b}''=1}^{2^{nR''}} 
 \frac{1}{2^{nR'(B-1)}}\sum_{m_1',\ldots,m_{B-1}'}
\nonumber\\ &\;
\sum_{y_{1,b-1}^n\in\Yset_1^n} \cprob{Y_{1,b-1}^n=y_{1,b-1}^n}{(M_1',M_1'')=(m_1',m_1''),\ldots,
(M_{b-1}',M_{b-1}'')=(m_{b-1}',m_{b-1}''),
S_1^n=s_1^n,\ldots, S_{b-1}^n=s_{b-1}^n } \nonumber\\ &\;\times
\sum_{y_{b}^n,y_{1,b}^n : g_b''(y_{b}^n,m_1',\ldots,m_{B-1}')\neq m_b'' } W_{Y^n|X^n,X_{1}^n,S^n}
(y_{b}^n| 
f_{b}(m_{b}',m_{b}''|m_{b-1}'),f_{1,b}(y_{1,b-1}^n),s_{b}^n  ) \,.
\label{eq:RYrch2}
\end{align}
The conditioning in the equations above can be explained as follows.
In (\ref{eq:RYrch1}), due to the code construction, the sequence $Y_{1,b}^n$ received at the relay in block
$b\in [1:B]$ depends only on the messages $(M_t',M_t'')$ with $t\leq b$. The decoded message $\hM_b'$, at the destination receiver, depends on messages $M_t'$ with $t>b$, since the receiver decodes this part of the message backwards. In (\ref{eq:RYrch2}), since the second part of the message $M_b''$ is decoded 
after backward decoding is complete, the estimation of $M_b''$ at the decoder depends on  the entire  sequence $\hM_1',\ldots,\hM_{B-1}'$.
%
By (\ref{eq:RYrCh1eps})--(\ref{eq:RYrCh2eps}), for every $t\in [1:b]$,  $h_b'$ and $h_b''$ as  functions of $s_{t+1}^n$ and $s_{t}^n$, respectively, satisfy (\ref{eq:RYRTcondC}) with $\alpha_n=e^{-2\theta n}$, given that the state sequences in the other blocks are fixed.
Hence, applying  Ahlswede's RT recursively, we obtain
\begin{align}
&\frac{1}{(n!)^{b+1}} \sum_{\pi_1,\pi_2,\ldots,\pi_{b+1}\in\Pi_n} h_b'(\pi_1 s_{1}^n, \pi_2 s_{2}^n,\ldots, 
\pi_{b+1} s_{b+1}^n)\leq 
(n+1)^{B|\Sset|}e^{-2\theta n} 
\leq e^{-\theta n}  \,,
 \,, \nonumber\\
&\frac{1}{(n!)^b} \sum_{\pi_1,\pi_2,\ldots,\pi_{b}\in\Pi_n} h_b''(\pi_1 s_{1}^n, \pi_2 s_{2}^n,\ldots, 
\pi_{b} s_{b}^n)\leq (n+1)^{B|\Sset|}e^{-2\theta n} 
\leq e^{-\theta n}  \,,
\label{eq:RYdetErrC}
\end{align} 
for all $(s_1^n,s_2^n,\ldots,s_{b+1}^n)\in\Sset^{(b+1)n}$ and sufficiently large $n$, such that $(n+1)^{B|\Sset|}\leq e^{\theta n}$.  

On the other hand, for every $\pi_1,\pi_2,\ldots,\pi_{b+1}\in\Pi_n$, we have that  
\begin{align}
h_b'(\pi_1 s_{1}^n, \pi_2 s_{2}^n,\ldots, 
\pi_{b+1} s_{b+1}^n)=\E\; h_b'(\pi_1 s_{1}^n, \pi_2 s_{2}^n,\ldots, 
\pi_{b+1} s_{b+1}^n|M_{t}',M_{t}'', t=1,\ldots,b+1) \,,
\end{align}
with
\begin{align}
&h_b'(\pi_1 s_{1}^n, \pi_2 s_{2}^n,\ldots, 
\pi_{b+1} s_{b+1}^n|m_{t}',m_{t}'', t=1,\ldots,b+1)																										\nonumber\\
=& \sum_{y_{1,1},\ldots,y_{1,b}}  \prod_{t=0}^{b-1} 
 W_{Y_1^n|X^n,X_1^n,S^n}(y_{1,t+1}^n|f_{t+1}(m_{t+1}',m_{t+1}''|m_t'),f_{1,t+1}(y_{1,t}^n),\pi_{t+1} s_{t+1}^n)
\nonumber\\
&\times\sum_{y_{b+1}^n : g_b'(y_{b+1}^n,m_{b+1}')\neq m_b' } W_{Y^n|X^n,X_{1}^n,S^n}(y_{b+1}^n| 
f_{b+1}(m_{b+1}',m_{b+1}''|m_b'),f_{1,b+1}(y_{1,b}^n),\pi_{b+1} s_{b+1}^n  )\nonumber\\
 \stackrel{(a)}{=}&
\sum_{y_{1,1},\ldots,y_{1,b}} \prod_{t=0}^{b-1}  
 W_{Y_1^n|X^n,X_1^n,S^n}(\pi_{t+1} y_{1,t+1}^n|f_{t+1}(m_{t+1}',m_{t+1}''|m_t'),f_{1,b+1}(\pi_t y_{1,t}^n),
\pi_{t+1} s_{t+1}^n)
\nonumber\\
&\times\sum_{y_{b+1}^n : g_b'(\pi_{b+1} y_{b+1}^n,m_{b+1}')\neq m_b' } W_{Y^n|X^n,X_{1}^n,S^n}(\pi_{b+1} y_{b+1}^n| 
f_{b+1}(m_{b+1}',m_{b+1}''|m_b'),f_{1,b+1}(\pi_b y_{1,b}^n),\pi_{b+1} s_{b+1}^n  )\nonumber\\
 \stackrel{(b)}{=}&
\sum_{y_{1,1},\ldots,y_{1,b}} 
 \prod_{t=0}^{b-1} 
 W_{Y_1^n|X^n,X_1^n,S^n}( y_{1,t+1}^n|\pi_{t+1}^{-1} f_{t+1}(m_{t+1}',m_{t+1}''|m_t'),\pi_{t+1}^{-1} f_{1,b+1}(\pi_t y_{1,t}^n), s_{t+1}^n)
\nonumber\\
&\times\sum_{y_{b+1}^n : g_b'(\pi_{b+1} y_{b+1}^n,m_{b+1}')\neq m_b' } W_{Y^n|X^n,X_{1}^n,S^n}
( y_{b+1}^n| \pi_{b+1}^{-1} f_{b+1}(m_{b+1}',m_{b+1}''|m_b'), \pi_{b+1}^{-1} f_{1,b+1}(\pi_b y_{1,b}^n), s_{b+1}^n  )  \,,
\label{eq:RYcerrpi1}
\end{align}
where $(a)$ is obtained by  
 changing the order of summation over $y_{1,1}^n,\ldots,y_{1,b}^n$ and $y_{b+1}^n$; and $(b)$ holds because the relay channel is memoryless.
Similarly,  
\begin{align}
h_b''(\pi_1 s_{1}^n, \pi_2 s_{2}^n,\ldots, \pi_{b} s_{b}^n)=
\E h_b''(\pi_1 s_{1}^n, \pi_2 s_{2}^n,\ldots, 
\pi_{b} s_{b}^n|M_{1}',\ldots,M_{B-1}', M_t'',t=1,\ldots,b) \,,
\end{align}
with
\begin{align}
&h_b''(\pi_1 s_{1}^n, \pi_2 s_{2}^n,\ldots, \pi_{b} s_{b}^n|m_{1}',\ldots,m_{B-1}', m_t'',t=1,\ldots,b)			\nonumber\\
=&  \sum_{y_{1,1},\ldots,y_{1,b-1}} \prod_{t=1}^{b-1}   
 W_{Y_1^n|X^n,X_1^n,S^n}(y_{1,t}^n|f_{t}(m_{t}',m_{t}''|m_{t-1}'),f_{1,t}(y_{1,t}^n),\pi_{t} s_{t}^n)
\nonumber\\
&\times
 \sum_{y_{b}^n : g_b''(y_{b}^n,m_1',\ldots,m_{B-1}')\neq m_b'' } W_{Y^n|X^n,X_1^n,S^n}(y_{b}^n| 
f_{b}(m_{b}',m_{b}''|m_{b-1}'),f_{1,b}(y_{1,b-1}^n),\pi_b s_{b}^n  )\nonumber\\
 \stackrel{(a)}{=}&
\sum_{y_{1,1},\ldots,y_{1,b-1}}  \prod_{t=1}^{b-1}   
 W_{Y_1^n|X^n,X_1^n,S^n}(\pi_t y_{1,t}^n|f_{t}(m_{t}',m_{t}''|m_{t-1}'),f_{1,t}(\pi_{t-1} y_{1,t-1}^n),\pi_{t} s_{t}^n)
\nonumber\\
&\times
 \sum_{y_{b}^n : g_b''(\pi_b y_{b}^n,m_1',\ldots,m_{B-1}')\neq m_b'' } W_{Y^n|X^n,X_1^n,S^n}(\pi_b y_{b}^n| 
f_{b}(m_{b}',m_{b}''|m_{b-1}'),f_{1,b}(\pi_{b-1} y_{1,b-1}^n),\pi_b s_{b}^n  )\nonumber\\
 \stackrel{(b)}{=}&
 \sum_{y_{1,1},\ldots,y_{1,b-1}}  \prod_{t=1}^{b-1}
 W_{Y_1^n|X^n,X_1^n,S^n}( y_{1,t}^n|\pi_t^{-1} f_{t}(m_{t}',m_{t}''|m_{t-1}'), \pi_t^{-1} f_{1,t}(\pi_{t-1} y_{1,t-1}^n),  s_{t}^n)
\nonumber\\
&\times
 \sum_{y_{b}^n : g_b''(\pi_b y_{b}^n,m_1',\ldots,m_{B-1}')\neq m_b'' } W_{Y^n|X^n,X_1^n,S^n}( y_{b}^n| 
\pi_b^{-1} f_{b}(m_{b}',m_{b}''|m_{b-1}'),\pi_b^{-1} f_{1,b}(\pi_{b-1} y_{1,b-1}^n), s_{b}^n  )  \,.
\label{eq:RYcerrpi2}
\end{align}
%
Then, consider the $(2^{nR(B-1)},nB)$ random Markov block code $\code_{BM}^\Pi$, specified by 
\begin{subequations}
\label{eq:RYCpi}
\begin{align}
&f_{b,\pi}(m_b',m_b''|m_{b-1}')= \pi_b^{-1} f_{b}(m_{b}',m_{b}''|m_{b-1}') \,,\quad
f_{1,b,\pi}(y_{1,b-1}^n)=\pi_b^{-1} f_{1,b}(\pi_{b-1} y_{1,b-1}^n) \,,
\intertext{and}
&g_{b,\pi}'(y_{b+1}^n,\hm_{b+1}')=g_b'(\pi_{b+1} y_{b+1}^n,\hm_{b+1}')
\,,\quad
 g_{b,\pi}''(y_{b}^n,\hm_{1}',\ldots,\hm_{B-1}')=g_b''(\pi y_{b}^n,\hm_1',\ldots,\hm_{B-1}') \,,
\end{align}
\end{subequations}
for $\pi_1,\ldots,\pi_B\in\Pi_n$,
with a uniform distribution $\mu(\pi_1,\ldots,\pi_B)=\frac{1}{|\Pi_n|^B}=\frac{1}{(n!)^B}$. 
That is, a set of $B$ independent permutations is chosen at random and applied to all blocks simultaneously, while the order of the blocks remains intact. 
As we restricted ourselves to a block Markov code, the relaying function in a given block depends only on symbols received in the previous block, hence, the relay can implement
those in-block permutations, and the coding scheme does not violate the causality requirement. 

 From (\ref{eq:RYcerrpi1}) and (\ref{eq:RYcerrpi2}), 
we see that using the random code $\code_{BM}^\Pi$, the error probabilities for the messages $M_b'$ and $M_b''$ are given by  
\begin{align}
&\text{Pr}_{\,\code_{BM}^\Pi}\hspace{-0.1cm} \left(\Eset_b' \,|\; (\Eset_{b+1}')^c, S_{1}^n=s_{1}^n,\ldots,  S_{b+1}^n=s_{b+1}^n \right) = \sum_{\pi_1,\ldots,\pi_B \in\Pi_n} \mu(\pi_1,\ldots,\pi_B) h_b'(\pi_1 s_{1}^n, \pi_2 s_{2}^n,\ldots, 
\pi_{b+1} s_{b+1}^n) \,,\nonumber\\
&\text{Pr}_{\,\code_{BM}^\Pi}\hspace{-0.1cm} \left(\Eset_b''\,|\;  \Eset_1'^c,\ldots,\Eset_{B-1}'^c,S_{1}^n=s_{1}^n,\ldots ,S_{b}^n=s_{b}^n  \right)  = \sum_{\pi_1,\ldots,\pi_B \in\Pi_n} \mu(\pi_1,\ldots,\pi_B) h_b''(\pi_1 s_{1}^n, \pi_2 s_{2}^n,\ldots, 
\pi_{b} s_{b}^n) \,,
\end{align}
for all $s_1^n,\ldots,s_{b+1}^n\in\Sset^n$, $b\in [1:B-1]$, and therefore, together with (\ref{eq:RYdetErrC}), we have that the probability of error of the random code $\code_{BM}^\Pi$ is bounded by 
$
\err(\qn,\code_{BM}^\Pi)\leq e^{-\theta n} 
$, 
for every $\qn(s^{nB})\in\pSpace(\Sset^{nB})$. That is, $\code_{BM}^\Pi$ is a $(2^{nR(B-1)},nB,e^{-\theta n})$ random 
 code for the AVRC $\avrc$, where the overall blocklength is $nB$, and the average rate $\frac{B-1}{B}\cdot R$ tends to $R$ as $B\rightarrow\infty$. 
This completes the proof of the partial decode-forward lower bound. 
\qed 

\subsection{Cutset Upper Bound}
The proof immediately follows from Lemma~\ref{lemm:RYcompoundDF},
since the random code capacity of the AVRC is bounded by the random code capacity  of the compound  relay channel, \ie
$
\RYrCav \leq \RYrCcompoundP 
$. 
\qed 

\section{Proof of Lemma~\ref{lemm:RYcorrTOdetC}}
\label{app:RYcorrTOdetC}
 We use the approach of \cite{Ahlswede:78p}, with the required adjustments. We use the random code constructed in the proof of  Theorem~\ref{theo:RYmain}. Let $R<\RYrCav$, and consider the case where the marginal sender-relay and sender-receiver AVCs have positive capacity, \ie
\begin{align}
\label{eq:RYrpos}
\opC(\avc_1(x_{1,1}))>0 \,,\;\text{and}\; \opC(\avc(x_{1,2}))>0 \,,
\end{align}
for some $x_{1,1},x_{1,2}\in\Xset_1$ (see (\ref{eq:RYmarginAVCs})).
By Theorem~\ref{theo:RYmain}, 
  for every $\eps>0$ and sufficiently large $n$, 
 there exists a $(2^{nR},n,\eps)$ random  code  
$
\code^\Gamma=\big(\mu(\gamma)=\frac{1}{k},\Gamma=[1:k],\{\code_\gamma \}_{\gamma\in \Gamma}\big) 
$, 
where
$\code_\gamma=(\encn_\gamma,\enc_{1,\gamma},\dec_{\gamma})$, for $\gamma\in\Gamma$. 
Following Ahlswede's Elimination Technique \cite{Ahlswede:78p}, it can be assumed that the size of the code collection is bounded by 
$k=|\Gamma|\leq n^2 $. 
By (\ref{eq:RYrpos}),  we have that for every $\eps'>0$ and sufficiently large $\nu'$, the code index $\gamma\in [1:k]$ can be sent through the relay channel $W_{Y_1|X,X_1,S}$ using a $(2^{\nu'\bR'},\nu',\eps')$ deterministic code 
$ 
\code_{\text{i}}'=(\tf^{\nu'},\gnu')  
$, where $\bR'>0$, while the relay repeatedly transmits the symbol $x_{1,1}$. 
Since $k$ is at most polynomial, 
 the encoder can reliably convey $\gamma$ to the relay with a negligible blocklength, \ie
$ 
\nu'=o(n) 
$. 
Similarly, there exists $(2^{\nu''\bR''},\nu'',\eps'')$ code $ 
\code_{\text{i}}''=(\tf^{\nu''},\gnu'')  
$ for the transmission of $\gamma\in [1:k]$  through the channel $W_{Y|X,X_1,S}$ to the receiver, where
$\nu''=o(n)$ and $\bR''>0$, while the relay repeatedly transmits the symbol $x_{1,2}$. 

Now, consider a code 
 formed by the concatenation of $\code_{\text{i}}'$ and $\code_{\text{i}}''$ as consecutive  prefixes to a corresponding code in the code collection $\{\code_\gamma\}_{\gamma\in\Gamma}$. 
That is, the encoder first sends the index $\gamma$ to the relay and the receiver, and then it sends the message $m\in [1:2^{nR}]$ to the receiver. Specifically, the encoder first transmits the $(\nu'+\nu'')$-sequence $(\tf^{\nu'}(\gamma),\tf^{\nu''}(\gamma))$ to convey the index $\gamma$, while the relay transmits the $(\nu'+\nu'')$-sequence $(\tx_1^{\nu'},\tx_1^{\nu''})$, where $\tx_1^{\nu'}=(x_{1,1},x_{1,1},\ldots,x_{1,1})$ and
$\tx_1^{\nu''}=(x_{1,2},x_{1,2},\ldots,x_{1,2})$. At the end of this transmission, the relay uses the first $\nu'$ symbols it received to  estimate the code index as $\hgamma'=\tg'(\ty_1^{\nu'})$. 

Then, the message $m$ is transmitted by the codeword $x^n=\enc_\gamma(m)$, while the relay transmits $x_1^n=\enc_{1,\hgamma'}^n(y_1^n)$. Subsequently, decoding is performed in two stages as well; the decoder estimates the index at first, with  
$\hgamma''=$ $\gnu''(\ty^{\nu''})$, and the message is then estimated by  
$\widehat{m}=$ $g_{\hgamma''}(y^n)$.  
 By the union of events bound, the probability of error 
 is then bounded by $\eps_c=\eps+\eps'+\eps''$, for every joint distribution in $\pSpace(\Sset^{\nu'+\nu''+n})$. 
That is, the concatenated code is a $(2^{(\nu'+\nu''+n)\tR_n},\nu'+\nu''+n,\eps_c)$ code over the AVRC $\avrc$, where 
  the blocklength is $n+o(n)$, and  
the rate   $\tR_{n}=\frac{n}{\nu'+\nu''+n}\cdot R$  approaches $R$ as $n\rightarrow \infty$. 
\qed

\section{Proof of Corollary~\ref{coro:RYmainDbound}}
\label{app:RYmainDbound}
Consider part 1.
By Definition~\ref{def:RYsymmetrizableGivenX1}, if
$W_{Y_1|X,X_1,S}$ and $W_{Y|X,X_1,S}$ are not symmetrizable -$\Xset|\Xset_1$ then there exist
$x_{1,1},x_{1,2}\in\Xset_1$ such that the DMCs $W_{Y_1|X,X_1,S}(\cdot|\cdot,$ $x_{1,1},$ $\cdot)$ and
$W_{Y|X,X_1,S}(\cdot|\cdot,x_{1,2},\cdot)$ are non-symmetrizable in the sense of \cite[Definition 2]{CsiszarNarayan:88p}.
This, in turn, implies that $\opC(\avc_1(x_{1,1}))>0$ and $\opC(\avc(x_{1,2}))>0$, 
due to \cite[Theorem 1]{CsiszarNarayan:88p}. Hence, by Lemma~\ref{lemm:RYcorrTOdetC},
$\RYCavc=\RYrCav$, and by Theorem~\ref{theo:RYmain}, $
\RYdIRavc \leq \RYCavc\leq\RYrICav$. 
%
 Parts 2 and 3 immediately follow from part 1 and Corollary~\ref{coro:RYavrcDeg}.
\qed

\section{Proof of Lemma~\ref{lemm:RYzeroCsymm}}
\label{app:RYzeroCsymm}
The proof is based on \cite{Ericson:85p}. 
Let $\avrc$ be a symmetrizable-$\Xset|\Xset_1$. Assume to the contrary that  a positive rate $R>0$ can be achieved. That is, for every $\eps>0$ and sufficiently large $n$, there exists a $(2^{nR},n,\eps)$
 code $\code=(f,f_1,g)$. Hence, the size of the message set is at least $2$, \ie
\begin{align}
\label{eq:RYzeroCmsize2}
\dM\triangleq 2^{nR} \geq 2 \,.
\end{align}  
We now show that there exists a distribution $q(s^n)$ such that the probability of error $\err(q,\code)$ is bounded from below by a positive constant, in contradiction to the assumption above.

By Definition~\ref{def:RYsymmetrizableGivenX1},  there exists a conditional distribution $J(s|x)$ that satisfies (\ref{eq:symmetrizableGivenX1}). Then, consider the 
 state sequence distribution
$
q(s^n)=\frac{1}{\dM} \sum_{m=1}^\dM J^n(s^n|x^n(m)) 
$, 
where $J^n(s^n|x^n)=\prod_{i=1}^n J(s_i|x_i)$ and $x^n(m)=f(m)$. For this distribution, the probability of error is given by
\begin{align}
\err(q,\code)
=& \sum_{s^n\in\Sset^n} \left[\frac{1}{\dM} \sum_{\tm=1}^\dM J^n(s^n|x^n(\tm)) \right]\cdot
\frac{1}{\dM} \sum_{m=1}^\dM \sum_{(y^n,y_1^n):g(y^n)\neq m } W^n(y^n,y_1^n|x^n(m),f_1^n(y_1^n),s^n) \nonumber\\
=& \frac{1}{2\dM^2} \sum_{m=1}^\dM \sum_{\tm=1}^\dM \sum_{(y^n,y_1^n):g(y^n)\neq m }\sum_{s^n\in\Sset^n}
W^n(y^n,y_1^n|x^n(m),f_1^n(y_1^n),s^n)J^n(s^n|x^n(\tm)) \nonumber\\
&+ \frac{1}{2\dM^2} \sum_{m=1}^\dM \sum_{\tm=1}^\dM \sum_{(y^n,y_1^n):g(y^n)\neq \tm }\sum_{s^n\in\Sset^n}
W^n(y^n,y_1^n|x^n(\tm),f_1^n(y_1^n),s^n)J^n(s^n|x^n(m))
\label{eq:RYzeroCsymm1}
\end{align}
with $W^n\equiv \nRC$ for short notation, where in the last sum we interchanged  the summation indices $m$ and $\tm$. Then, consider the last sum, and observe that by (\ref{eq:symmetrizableGivenX1}), we have that
\begin{align}
\sum_{s^n\in\Sset^n} W^n(y^n,y_1^n|x^n(\tm),f_1^n(y_1^n),s^n)J^n(s^n|x^n(m))
=& \prod_{i=1}^n \left[\sum_{s_i\in\Sset} W(y_i,y_{1,i}|x_i(\tm),f_{1,i}(y_1^{i-1}),s_i) J(s_i|x_i(m))  \right]
\nonumber\\
=& \prod_{i=1}^n \left[\sum_{s_i\in\Sset} W(y_i,y_{1,i}|x_i(m),f_{1,i}(y_1^{i-1}),s_i) J(s_i|x_i(\tm))  \right]
\nonumber\\
=& \sum_{s^n\in\Sset^n} W^n(y^n,y_1^n|x^n(m),f_1^n(y_1^n),s^n)J^n(s^n|x^n(\tm)) \,.
\label{eq:RYzeroCsymm2}
\end{align}
Substituting (\ref{eq:RYzeroCsymm2}) in (\ref{eq:RYzeroCsymm1}), we have 
\begin{align}
\err(q,\code) 
=& \frac{1}{2\dM^2} \sum_{m=1}^\dM \sum_{\tm=1}^\dM \sum_{s^n\in\Sset^n} \bigg[
 \sum_{(y^n,y_1^n):g(y^n)\neq m }
W^n(y^n,y_1^n|x^n(m),f_1^n(y_1^n),s^n)J^n(s^n|x^n(\tm)) \nonumber\\
&+  \sum_{(y^n,y_1^n):g(y^n)\neq \tm } W^n(y^n,y_1^n|x^n(m),f_1^n(y_1^n),s^n)J^n(s^n|x^n(\tm)) \bigg]
\nonumber\\
\geq& \frac{1}{2\dM^2} \sum_{m=1}^\dM \sum_{\tm\neq m} \sum_{s^n\in\Sset^n} 
 \sum_{y^n,y_1^n }
W^n(y^n,y_1^n|x^n(m),f_1^n(y_1^n),s^n)J^n(s^n|x^n(\tm)) \nonumber\\
=& \frac{\dM (\dM-1)}{2\dM^2} \geq \frac{1}{4} \,,
\label{eq:RYzeroCsymm3}
\end{align}
where the last inequality follows from (\ref{eq:RYzeroCmsize2}), hence a positive rate cannot be achieved.
\qed

\section{Proof of Lemma~\ref{lemm:RYzeroCsymmDeg}}
\label{app:RYzeroCsymmDeg}
Let $\avrc=\{W_{Y_1|X,X_1} W_{Y|Y_1,X_1,S} \}$ be a symmetrizable-$\Xset_1\times \Yset_1$ degraded AVRC. Assume to the contrary that  a positive rate $R>0$ can be achieved. That is, for every $\eps>0$ and sufficiently large $n$, there exists a $(2^{nR},n,\eps)$
 code $\code=(f,f_1,g)$. Hence, the size of the message set is at least $2$, \ie
\begin{align}
\label{eq:RYzeroCmsize2d}
\dM\triangleq 2^{nR} \geq 2 \,.
\end{align}  
We now show that there exists a distribution $q(s^n)$ such that the probability of error $\err(q,\code)$ is bounded from below by a positive constant, in contradiction to the assumption above.
By Definition~\ref{def:RYsymmetrizableGivenY1},  there exists a conditional distribution $J(s|x_1,y_1)$ that satisfies (\ref{eq:symmetrizableGivenY1}). Then, consider the following state sequence distribution,
\begin{align}
q(s^n)=\frac{1}{\dM} \sum_{m=1}^\dM \sum_{y_1^n\in\Yset_1} W_{Y_1^n|X^n,X_1^n}(y_1^n|f(m),f_1^n(y_1^n))
 J^n(s^n|f_1^n(y_1^n),y_1^n) \,,
\end{align}
where $J^n(s^n|x_1^n,y_1^n)=\prod_{i=1}^n J(s_i|x_{1,i},y_{1,i})$. For this distribution, the probability of error is given by
\begin{align}
\err(q,\code)
=& \sum_{s^n\in\Sset^n} \left[\frac{1}{\dM} \sum_{\tm=1}^\dM \sum_{\ty_1^n} W_{Y_1^n|X^n,X_1^n}(\ty_1^n|f(\tm),f_1^n(\ty_1^n))
 J^n(s^n|f_1^n(\ty_1^n),\ty_1^n) \right] \nonumber\\ &\times
\frac{1}{\dM} \sum_{m=1}^\dM \sum_{(y^n,y_1^n):g(y^n)\neq m } W_{Y_1^n|X^n,X_1^n}(y_1^n|f(m),f_1^n(y_1^n))
W^n(y^n|y_1^n,f_1^n(y_1^n),s^n) 
 \nonumber\\
=& \frac{1}{2\dM^2} \sum_{m=1}^\dM \sum_{\tm=1}^\dM \sum_{y_1^n,\ty_1^n}W_{Y_1^n|X^n,X_1^n}(\ty_1^n|f(\tm),f_1^n(\ty_1^n))\cdot W_{Y_1^n|X^n,X_1^n}(y_1^n|f(m),f_1^n(y_1^n)) \nonumber\\ &\times
\sum_{y^n:g(y^n)\neq m }\sum_{s^n\in\Sset^n}
W^n(y^n|y_1^n,f_1^n(y_1^n),s^n)J^n(s^n|f_1^n(\ty_1^n),\ty_1^n) \nonumber\\
&+ \frac{1}{2\dM^2} \sum_{m=1}^\dM \sum_{\tm=1}^\dM \sum_{y_1^n,\ty_1^n}  W_{Y_1^n|X^n,X_1^n}(y_1^n|f(m),f_1^n(y_1^n))\cdot W_{Y_1^n|X^n,X_1^n}(\ty_1^n|f(\tm),f_1^n(\ty_1^n)) \nonumber\\ &\times
\sum_{y^n:g(y^n)\neq \tm }\sum_{s^n\in\Sset^n}
W^n(y^n|\ty_1^n,f_1^n(\ty_1^n),s^n)J^n(s^n|f_1^n(y_1^n),y_1^n)
\label{eq:RYzeroCsymm1d}
\end{align}
with $W^n\equiv W_{Y^n|Y_1^n,X_1^n,S^n}$ for short notation, where in the last sum we interchanged  the summation variables $(m,y_1^n)$ and $(\tm,\ty_1^n)$. Then, consider the last sum, and observe that by (\ref{eq:symmetrizableGivenY1}), we have that
\begin{align}
\sum_{s^n\in\Sset^n} W^n(y^n|\ty_1^n,f_1^n(\ty_1^n),s^n)J^n(s^n|f_1^n(y_1^n),y_1^n)
=& \prod_{i=1}^n \left[\sum_{s_i\in\Sset} W(y_i|\ty_{1,i},f_{1,i}(\ty_1^{i-1}),s_i) 
J(s_i|f_{1,i}(y_1^{i-1}),y_{1,i})  \right]
\nonumber\\
=& \prod_{i=1}^n \left[\sum_{s_i\in\Sset} W(y_i|y_{1,i},f_{1,i}(y_1^{i-1}),s_i) 
J(s_i|f_{1,i}(\ty_1^{i-1}),\ty_{1,i})  \right]
\nonumber\\
=& \sum_{s^n\in\Sset^n} W^n(y^n|y_1^n,f_1^n(y_1^n),s^n)J^n(s^n|f_1^n(\ty_1^n),\ty_1^n) \,.
\label{eq:RYzeroCsymm2d}
\end{align}
Substituting (\ref{eq:RYzeroCsymm2d}) in (\ref{eq:RYzeroCsymm1d}), we have 
\begin{align}
&\err(q,\code) 
= \frac{1}{2\dM^2} \sum_{m=1}^\dM \sum_{\tm=1}^\dM  \sum_{y_1^n,\ty_1^n}  W_{Y_1^n|X^n,X_1^n}(\ty_1^n|f(\tm),f_1^n(\ty_1^n)) W_{Y_1^n|X^n,X_1^n}(y_1^n|f(m),f_1^n(y_1^n))
\nonumber\\&\times \sum_{s^n\in\Sset^n} \bigg[
\sum_{y^n:g(y^n)\neq m }
W^n(y^n|y_1^n,f_1^n(y_1^n),s^n)J^n(s^n|f_1^n(\ty_1^n),\ty_1^n) 
+ \sum_{y^n:g(y^n)\neq \tm }
W^n(y^n|y_1^n,f_1^n(y_1^n),s^n)J^n(s^n|f_1^n(\ty_1^n),\ty_1^n) \bigg]
\nonumber\\
\geq& \frac{1}{2\dM^2} \sum_{m=1}^\dM \sum_{\tm\neq m}  \sum_{y_1^n,\ty_1^n}  W^n
(\ty_1^n|f(\tm),f_1^n(\ty_1^n)) W^n
(y_1^n|f(m),f_1^n(y_1^n))
 \sum_{s^n\in\Sset^n} \sum_{y^n\in\Yset^n }
W^n(y^n|y_1^n,f_1^n(y_1^n),s^n)J^n(s^n|f_1^n(\ty_1^n),\ty_1^n) 
\nonumber\\
=& \frac{\dM (\dM-1)}{2\dM^2} \geq \frac{1}{4} \,,
\label{eq:RYzeroCsymm3d}
\end{align}
where the last inequality follows from (\ref{eq:RYzeroCmsize2d}), hence a positive rate cannot be achieved.
\qed

\section{Analysis of Example 
 \ref{example:BSRCsymm}}
\label{app:BSRCsymm}
We show that the random code capacity of the AVRC in Example~\ref{example:BSRCsymm} 
is given by  $\RYrCav=\min\left\{ \frac{1}{2},1-h(\theta)  \right\}$.
As the AVRC is degraded, 
 the random code capacity is given by 
\begin{align}
\RYrCav= \RYdIRavc=\RYrICav=  \max_{p(x,x_1)}   \min \left\{ \min_{0\leq q\leq 1} I_q(X,X_1;Y) \,,\; 
 I(X;Y_1|X_1)
\right\} \,,
\end{align}
due to part 2 of Corollary~\ref{coro:RYavrcDeg}, where $q\equiv q(1)=1-q(0)$. 
Now, consider the direct part.
Set $p(x,x_1)=p(x)p(x_1)$, where
 $X\sim \text{Bernoulli}(\nicefrac{1}{2})$ and $X_1\sim \text{Bernoulli}(\nicefrac{1}{2})$. Then,
\begin{align}
&I(X;Y_1|X_1)
=1-h(\theta) \,, \nonumber\\
&H_q(Y)= \frac{1}{2}\left[ -q \log\left( \frac{1}{2} q \right)-(1-q) \log\left( \frac{1}{2} (1- q) \right) \right] -\frac{1}{2} \log\left( \frac{1}{2}  \right)
=1+\frac{1}{2} h(q)
 \,, \nonumber\\
&H_q(Y|X,X_1)=h(q) \,.
\end{align}
Hence,
\begin{align}
\RYrCav \geq& \min\left\{\min_{0\leq q\leq 1}\left[ 1-\frac{1}{2} h(q)  \right],1-h(\theta)\right\}=
\min\left\{\frac{1}{2},1-h(\theta) \right\} \,.
\intertext{
As for the converse part, we have the following bounds,
}
\RYrCav\leq&  \max_{p(x,x_1)} I(X;Y_1|X_1)=1-h(\theta) \,,
\intertext{and}
\RYrCav\leq&  \max_{p(x,x_1)}    \min_{0\leq q\leq 1} I_q(X,X_1;Y) 
\leq 
 \max_{p(x,x_1)}   [ H_q(Y)-H_q(Y|X,X_1) ] \Big|_{q=\frac{1}{2}}
\nonumber\\
=& \max_{0\leq p\leq 1}     \left[1+\frac{1}{2} h(p) \right]-1=\frac{1}{2} \,,
\end{align}
where $p\triangleq \prob{X_1 =1}$. 
\qed

\section{Proof of Lemma~\ref{lemm:RYmainOrS}}
\label{app:RYmainOrS}
The proof follows the lines of \cite{ElGamalZahedi:05p}.
Consider an AVRC $\avrc$ $=$ $\{  W_{Y|X',X_1}$ $ W_{Y_1|X'',X_1,S} \}$ with orthogonal sender components.
We apply Theorem~\ref{theo:RYmain}, which states that $\RYdIRavc\leq \RYrCav\leq \RYrICav$. 

\subsection{Achievability Proof}
To show achievability,  we set
$U=X''$ and $p(x',x'',x_1)=p(x_1)p(x'|x_1)p(x''|x_1)$ in the partial decode-forward lower bound $\RYdIRavc \triangleq \inR_{DF}(\RYcompound)
\bigg|_{\Qset=\pSpace(\Sset)}$. Hence, by (\ref{eq:RYdIRcompound}), 
\begin{align}
\dRYdIRavc\geq& \max_{p(x_1)p(x'|x_1)p(x''|x_1)} \min \Big\{   I(X',X'',X_1;Y) \,,\; 
\min_{q(s)}
 I_{q} (X'';Y_1|X_1)+  I(X';Y|X_1,X'') 
\Big\} \,.
\label{eq:OrSlowerB1}
\end{align}
Now, by (\ref{eq:OrSrc}), we have that $(X'',Y_1)\Cbar (X',X_1)\Cbar Y$ form a Markov chain.
As $(X_1,X',X'')\sim p(x_1)p(x'|x_1)p(x''|x_1)$, it further follows that $(X'',Y_1)\Cbar X_1\Cbar Y$ form a Markov chain, hence $I(X',X'',X_1;Y)=I(X',X_1;Y)$ and $ I(X';Y|X_1,X'')= I(X';Y|X_1)$.
Thus, (\ref{eq:OrSlowerB1}) reduces to the expression in the RHS of (\ref{eq:RYICavcOrthoS}).
If 
$W_{Y_1|X'',X_1,S}$ is non-symmetrizable-$\Xset''|\Xset_1$,  then (\ref{eq:OrSlowerB1}) is achievable by deterministic codes as well, due to Corollary~\ref{coro:RYmainDbound}.
\qed

\subsection{Converse Proof}
By (\ref{eq:RYICcompound}) and (\ref{eq:RYIRcompoundP}), the cutset upper bound takes the following form,
\begin{align}
\RYrICav
=&\min_{q(s)} \max_{p(x',x'',x_1)} 
\min \big\{ I(X',X'',X_1;Y) \,,\;   I_{q}(X',X'';Y,Y_1|X_1)
\big\} 
\nonumber\\
=& \max_{p(x',x'',x_1)} 
\min \big\{ I(X',X'',X_1;Y) \,,\;  \min_{q(s)} I_{q}(X',X'';Y,Y_1|X_1)
\big\} \,,
\label{eq:RYorsUpperb1}
\end{align}
where the last line is due to the minimax theorem \cite{sion:58p}. 
For the 
AVRC with orthogonal sender components, as specified by (\ref{eq:OrSrc}), we have the following Markov relations,
\begin{align}
&Y_1\Cbar (X'',X_1)\Cbar (X',Y) \,, \label{eq:RYorsY1markov}\\
&(X'',Y_1)\Cbar (X',X_1)\Cbar Y \,. \label{eq:RYorsYmarkov} 
\end{align}
Hence,  by (\ref{eq:RYorsYmarkov}), 
$
I(X',X'',X_1;Y)=I(X',X_1;Y) 
$. 
 As for the second mutual information in the RHS of (\ref{eq:RYorsUpperb1}), by the mutual information chain rule,
\begin{align}
I_{q}(X',X'';Y,Y_1|X_1)=& I_{q}(X'';Y_1|X_1)+I_{q}(X';Y_1|X'',X_1)+I_{q}(X',X'';Y|X_1,Y_1)
\nonumber\\
\stackrel{(a)}{=}& I_{q}(X'';Y_1|X_1)+I_{q}(X',X'';Y|X_1,Y_1)
\nonumber\\
\stackrel{(b)}{=}& I_{q}(X'';Y_1|X_1)+H_{q}(Y|X_1,Y_1)-H(Y|X',X_1)
\nonumber\\
\stackrel{(c)}{\leq}& I_{q}(X'';Y_1|X_1)+I(X';Y|X_1)
\end{align}
where $(a)$ is due to (\ref{eq:RYorsY1markov}), $(b)$ is due to (\ref{eq:RYorsYmarkov}),  and
$(c)$ holds since conditioning reduces entropy. Therefore,
\begin{align}
\RYrICav\leq    \max_{p(x',x'',x_1)} 
\min \left\{ I(X',X_1;Y) \,,\;  \min_{q(s)} I_{q}(X'';Y_1|X_1)+I(X';Y|X_1)
\right\} .
\label{eq:RYorsUpperb2}
\end{align}
Without loss of generality, the maximization in (\ref{eq:RYorsUpperb2}) can be restricted to
distributions of the form $p(x',x'',x_1)=p(x_1)\cdot$ $p(x'|x_1)\cdot$ $p(x''|x_1)$.
\qed

\end{appendices}

\printbibliography
 
\end{document}